\documentclass[10pt]{article}
\usepackage{cancel}
\usepackage[T1]{fontenc}
\usepackage[utf8]{inputenc}
\usepackage{times}
\usepackage[margin=1in]{geometry}
\usepackage{amsmath,amssymb,amsthm,mathtools}
\usepackage{bm}
\usepackage{xcolor}
\usepackage{thm-restate}
\usepackage[colorlinks=true,linkcolor=blue,citecolor=blue,urlcolor=blue]{hyperref}
\usepackage{bbm,booktabs,tabularx,array}

\DeclarePairedDelimiter{\bra}{\langle}{\rvert}%
\DeclarePairedDelimiter{\ket}{\lvert}{\rangle}%
\DeclarePairedDelimiterX\innerp[2]{\langle}{\rangle}{#1\delimsize\vert\mathopen{}#2}%
\DeclarePairedDelimiterX\braket[2]{\langle}{\rangle}{#1\delimsize\vert\mathopen{}#2}%
\DeclarePairedDelimiterX\braketOP[3]{\langle}{\rangle}{#1\,\delimsize\vert\,\mathopen{}#2\,\delimsize\vert\,\mathopen{}#3}%
\DeclarePairedDelimiterX\ketbra[2]{\lvert}{\rvert}{#1\delimsize\rangle\!\delimsize\langle#2}%
\DeclarePairedDelimiterX\outerp[2]{\lvert}{\rvert}{#1\delimsize\rangle\!\delimsize\langle#2}%
\DeclarePairedDelimiterX\projector[1]{\lvert}{\rvert}{#1\delimsize\rangle\!\delimsize\langle#1}%

\DeclarePairedDelimiter\parens{\lparen}{\rparen}
\DeclarePairedDelimiter\abs{\lvert}{\rvert}
\DeclarePairedDelimiter\norm{\lVert}{\rVert}

\DeclarePairedDelimiterX\diverg[2]{(}{)}{#1 \mathrel{}\mathclose{}\delimsize\|\mathopen{}\mathrel{} #2}
\newcommand{\dtv}[2]{\mathrm{d}_{\mathrm{TV}}(#1,#2)}

\newcommand{\dKL}[2]{\mathrm{d}_{\mathrm{KL}}\diverg{#1}{#2}}

\newcommand{\DKLmeas}[2]{\mathrm{D}_{\mathrm{KL}}^{\textnormal{meas}}\diverg{#1}{#2}}

\newcommand{\eps}{\epsilon} 
\newcommand{\calD}{\mathcal{D}}
\newcommand{\wt}[1]{\widetilde{#1}}

\renewcommand{\epsilon}{\varepsilon}

\usepackage[ruled,vlined,linesnumbered]{algorithm2e}

\mathtoolsset{showonlyrefs}

\let\mathbm\bm

\newcommand{\energy}{\mathcal{E}}
\newcommand{\luders}{\Phi}

\newcommand{\ba}{\mathbm{a}}
\newcommand{\bj}{\mathbm{j}}
\newcommand{\brho}{\mathbm{\rho}}
\newcommand{\bmu}{\mathbm{\mu}}

\newcommand{\CC}{\mathbb{C}}

\DeclareMathOperator*{\EE}{\mathbb{E}}

\renewcommand{\bm}{\mathbf{m}}

\newcommand{\calH}{\mathcal{H}}
\newcommand{\calR}{\mathcal{R}}

\DeclareMathOperator{\Tr}{Tr}

\DeclareMathOperator{\poly}{poly}

\newcommand{\Id}{\mathbbm{1}}

\newtheorem{theorem}{Theorem}[section]
\newtheorem{lemma}[theorem]{Lemma}
\newtheorem{definition}[theorem]{Definition}

\newtheorem{claim}[theorem]{Claim}
\newtheorem{corollary}[theorem]{Corollary}
\newtheorem{remark}[theorem]{Remark}

\title{Online Shadow Tomography Matching the Classical Bounds}
\author{
Sitan Chen
\thanks{SEAS, Harvard University. Email: \href{mailto:sitan@seas.harvard.edu}{sitan@seas.harvard.edu}.} \qquad
Ryan O'Donnell
\thanks{Computer Science Department, Carnegie Mellon University. Email: \href{mailto:odonnell@cs.cmu.edu}{odonnell@cs.cmu.edu}.} \qquad
Angelos Pelecanos
\thanks{UC Berkeley. Email: \href{mailto:apelecan@berkeley.edu}{apelecan@berkeley.edu}.} \qquad
John Wright
\thanks{UC Berkeley. Email: \href{mailto:jswright@berkeley.edu}{jswright@berkeley.edu}.}
}

\newcommand{\calT}{\mathcal{T}}
\newcommand{\sitan}[1]{{\color{red}[sitan: #1]}}
\newcommand{\angelos}[1]{{\color{blue}[angelos: #1]}}
\newcommand{\ryan}[1]{{\color{purple}[ryan: #1]}}
\newcommand{\john}[1]{{\color{orange}[john: #1]}}

\renewcommand{\sitan}[1]{}
\renewcommand{\angelos}[1]{}
\renewcommand{\ryan}[1]{}
\renewcommand{\john}[1]{}

\newcommand{\lift}[1]{\overline{#1}}

\begin{document}

\maketitle
\begin{abstract}
    In \emph{Online Shadow Tomography}, we are given copies of an unknown $d$-dimensional quantum state $\rho$, an adversary (adaptively) proposes a sequence of bounded observables $A^{(1)},\ldots,A^{(m)}$, and after each $A^{(t)}$ is given we must estimate $\Tr(A^{(t)}\rho)$ to within $\pm \epsilon$. 
    This is the direct quantum generalization of the classical problem of \emph{Adaptive Data Analysis}. %The ``offline'' case, in which $A^{(1)}, \ldots, A^{(m)}$ are given upfront, is also a well-studied problem.  The main goal is to minimize the number of copies, $n$, required.
    %The best prior results for \emph{offline} Shadow Tomography were $n = \log^2(m) \cdot \wt{\mathcal{O}}(\log(d)/\eps^4)$ due to~\cite{BO24,BB24} (and also working online), and $n = \wt{\mathcal{O}}(\sqrt{m} \log(m)/\eps^2)$ due to~\cite{Sin25}.
    %In contrast, the best prior classical online bounds are $\mathcal{O}(\log(m) \sqrt{\log d}/\eps^3)$ due to~\cite{bassily2021algorithmic,CLNSS23} and $\mathcal{O}(\log(m) \sqrt{\log d}/\eps^3)$.  Another prior result achieved $n = \mathcal{O}(\sqrt{m} \log(m)/\eps^2)$ but was not online, unlike the analogous classical result~\cite{bassily2021algorithmic}.
    Prior results for online Shadow Tomography were suboptimal in all three parameters $m, d, \epsilon$, lagging behind the best known and classical rates~\cite{bassily2021algorithmic}, for which there is some evidence of optimality~\cite{NSSSU18,lyu2025fingerprinting}. In this work, we finally close this gap, giving a pair of algorithms achieving  sample complexities of
    \begin{equation}
        \mathcal{O}\parens*{\frac{\log(m) \sqrt{\log d}}{\epsilon^3}}\qquad \text{and} \qquad \mathcal{O}\parens*{\frac{\sqrt{m}}{\epsilon^2}}\,,
    \end{equation}
    matching the classical rates. 
    The bound on the left is the first to achieve $o(\log^2 m)$-dependence together with $\poly(\log(d)/\eps)$; moreover, it improves all three exponents even in the \emph{Offline} Shadow Tomography setting.  The bound on the right is known to be optimal among bounds independent of~$d$, and improves the best prior result by a $\sqrt{m} \log m$ factor.
    
    The key to our proof is a new framework for quantifying post-measurement damage, based on the quantum Efron--Stein decomposition.% recently developed by~\cite{de2025non}.
\end{abstract}

\newpage

\tableofcontents

\newpage

% {\color{red}
% Notational conventions:
%     \begin{itemize}
%         \item $A^{(t)}$: $t$-th observable
%         \item $A_j = \Id^{\otimes j-1} \otimes A \otimes \Id^{n - j}$
%         \item $n$: number of copies
%         \item $\mathsf{N}$: number operator
%         \item ``high'' and ``low''
%         \item ``bad''
%         \item Indicator of quantum event $E$: $\mathbf{E}$
%     \end{itemize}
% }

\section{Introduction}

In this work we revisit the well-studied problem of \emph{Shadow Tomography}~\cite{Aar16,Aar20}.

\begin{definition}[Shadow Tomography]
    Let $\rho$ be an unknown state on a $d$-dimensional Hilbert space $\calH$. In \emph{Online Shadow Tomography}, we are given $n$ copies of $\rho$, and subsequently receive a sequence of adaptively chosen observables $0\preceq A^{(1)},\ldots,A^{(m)}\preceq \Id$. After receiving $A^{(t)}$, we must respond with an answer $\widehat{\mu}^{(t)}$ satisfying 
    \begin{equation}
        |\widehat{\mu}^{(t)} - \mu^{(t)}| \le \epsilon\,, \qquad \mathrm{for} \qquad \mu^{(t)} \triangleq \Tr(A^{(t)}\rho)
    \end{equation}
    for some error parameter $\epsilon > 0$, and the adversary selects the next observable $A^{(t+1)}$ based on the transcript $\calT_t=\{(A^{(j)},\widehat{\mu}^{(j)})\}_{j\le t}$ of the interaction thus far. \emph{Offline Shadow Tomography} is the special case where the $A^{(t)}$'s are all fixed ahead of time.
\end{definition}

\noindent In the special case where $\rho$ and all $A^{(t)}$'s are diagonal matrices, this reduces to the classical question of \emph{Adaptive Data Analysis}, for which the best known upper bounds on the sample complexity $n$ scale as the minimum of
\begin{equation}
    n = \mathcal{O}\Bigl(\min\Bigl(\frac{\log(m)\sqrt{\log d}}{\epsilon^3}, \frac{\sqrt{m}}{\epsilon^2}\Bigr)\Bigr)\,, \label{eq:rate}
\end{equation}
and there is evidence that this is optimal~\cite{NSSSU18,lyu2025fingerprinting}.
In contrast, in the quantum setting, the best known upper bounds~\cite{BO24,BB24} to hold in all parameter regimes scale as
\begin{equation}
    n = \mathcal{O}\Bigl(\frac{\log^2(m)\log d}{\epsilon^4}, \frac{m}{\epsilon^2}\Bigr)\,.\label{eq:old_rate}
\end{equation}
In the offline case, recent work~\cite{Sin25} improved the latter bound to $\mathcal{O}(\sqrt{m} \log(m)/\epsilon^2)$.

This state of affairs is unsatisfying on several counts. In the online setting, if $d$ is exponentially large in $m$, then the best known rate is the trivial estimator. Even in the offline setting, when the number of observables $m$ is exponentially large in the dimension $d$, the best known rate is sub-optimal in every parameter. In this work, we finally resolve these issues, closing a long-standing gap between the classical and quantum rates.

\begin{restatable}{theorem}{logm}\label{thm:main_logm}
    There is a protocol for Online Shadow Tomography of $d$-dimensional states that, using $n$ copies of the unknown state for
    \begin{equation}
        n = \mathcal{O}\Bigl(\frac{\sqrt{K}\log(m+K)}{\epsilon^2}\Bigr) \qquad \mathrm{for} \qquad K = \Theta(\log(d)/\epsilon^2)\,,
    \end{equation}
    correctly answers $m$ adaptively chosen observables to error $\epsilon$ with probability at least $9/10$.
\end{restatable}

\begin{restatable}{theorem}{sqrtm}\label{thm:main_sqrtm}
    There is a protocol for Online Shadow Tomography of $d$-dimensional states that, using $n$ copies of the unknown state for
    \begin{equation}
        n = \mathcal{O}\Bigl(\frac{\sqrt{m}}{\epsilon^2}\Bigr)\,,
    \end{equation}
    correctly answers $m$ adaptively chosen observables to error $\epsilon$ with probability at least $9/10$. In particular, the sample complexity of this protocol is independent of the dimension $d$.
\end{restatable}

\noindent Elementary casework shows that the minimum of these two rates is equal to the classical rate in Eq.~\eqref{eq:rate} up to constants. 

We achieve both of these results through a new method for accounting for post-measurement damage through the quantum Efron--Stein decomposition~\cite{pickl2011simple,de2025non}. Interestingly, when specialized to the classical setting, this appears to give new proofs for the classical rates that do not go through Differential Privacy but instead through Fourier analysis of functions over product spaces.

\subsection{Related work}

\paragraph{Adaptive Data Analysis.}  The classical version of Online Shadow Tomography is the \emph{Adaptive Data Analysis} problem, introduced by Dwork et~al.~\cite{DFHPRR16}.  In this problem, we are given samples from an unknown probability distribution $p$ on $[d]$, an adversary adaptively proposes a sequence of random variables (``queries'') 
% \sitan{was using $a^{(t)}$ for ``answer'', maybe we can switch these to $q^{(t)}$?}\ryan{ooh, that's a bit rough if we make the change from $O$ to $A$ for observables; I kinda like having classical $a$ to parallel quantum $A$; my dream replacement for your $a^{(t)}$ would be $\widehat{\mu}^{(t)}$} \sitan{$\widehat{\mu}^{(t)}$ sounds good! changed just now}
$a^{(1)}, \ldots, a^{(m)} : [d] \to [0,1]$, and after each $a^{(t)}$ is given we must output an estimate of $\EE_p[a]$ to within~$\pm \eps$.  (For simplicity, we fix the allowed failure probability $\delta$ of any estimate being inaccurate to~$1/3$.)  Using tools from Differential Privacy,
% \ryan{maybe explain a bit more, saying gentle/damage}, 
Dwork et~al.\ showed that $n = \wt{\mathcal{O}}(\log^{3/2}(m) \sqrt{\log d}/\eps^{7/2})$ or $n = \wt{\mathcal{O}}(\log(m) \log(d)/\eps^{4})$ samples suffice, and these were improved to just $\mathcal{O}(\log(m) \sqrt{\log d}/\eps^3)$ by Bassily et~al.~\cite{bassily2021algorithmic} (see~\cite{CLNSS23} for an improved dependence on~$\delta$). These tools are closely related to the essence of Shadow Tomography: the premise of Differential Privacy is to perform data analysis that is as insensitive as possible to individual-level details in the dataset, which can be regarded as a classical analogue of the quantum notion of performing ``gentle'' measurements that do not damage the quantum state we are trying to learn about.

One can also seek \emph{dimension-independent} bounds, meaning ones with no dependence on~$d$.  In this case, Bassily et~al.\ showed that $n = \mathcal{O}(\sqrt{m \log m} \log(1/\eps)/\eps^2)$  or $n = \mathcal{O}(\sqrt{m \log \log m} \log^{3/2}(1/\eps)/\eps^2)$ suffices, and Dagan--Kur~\cite{DK22} improved this to $n = \mathcal{O}(\sqrt{m \log(1/\eps)}/\eps^2)$ provided $\eps \geq \exp(-\wt{\Omega}(m))$. Interestingly, as a byproduct our Theorem~\ref{thm:main_sqrtm} appears to improve further upon this by removing the $\sqrt{\log(1/\epsilon)}$ factor, albeit only in the regime of constant failure probability.

Let us now mention lower bounds (all of which automatically also apply to the harder problem of online Shadow Tomography).  Trivially, $\Omega(\log(m)/\eps^2)$ is a lower bound, once $d = \Omega(\log m)$.  For dimension-independent results, \cite{NSSSU18} show that $n = \Omega(\sqrt{m}/\eps^2)$ is a lower bound provided $d \geq \exp(O(m))$, matching the upper bound, provided the algorithm is $\eps$-accurate not only on the underlying distribution, but also on the $n$ empirical samples.  In general, Lyu--Talwar~\cite{lyu2025fingerprinting} show that $\Omega(\log(m)\sqrt{\log d}/(\eps^2 \log(1/\eps)))$ is a lower bound (assuming $\log d \ll m \ll 2^d$ and $\eps \gg 1/d, 1/m$); they also give evidence that the known upper bound of $\mathcal{O}(\log(m) \sqrt{\log d}/\eps^3)$ is essentially tight by showing that $\Omega(\log(m) \sqrt{\log d}/(\eps^3 \log(1/\eps)))$ is a lower bound for algorithms, again, which are $\eps$-accurate not just on the distribution but also on the empirical samples.

\paragraph{Shadow Tomography.}
The offline form of Shadow Tomography was introduced by Aaronson~\cite{Aar16} in 2016, and the online form first discussed by Aaronson and Rothblum~\cite{aaronson2019gentle}.  In the online case, upper bounds of $n = \mathcal{O}(m\log(m)/\eps^2)$ and $n = \mathcal{O}(d^2/\eps^2)$ are trivial (the last using full tomography~\cite{OW16,HHJWY17}).  Aaronson and Rothblum showed that $n = \mathcal{O}(\log^2(m) \log^2(d)/\eps^8)$ suffices.  This was later improved by B\u{a}descu--O'Donnell~\cite{BO24} to $n = \log^2(m) \cdot \mathcal{O}(\log(d)/\eps^4)$. 
The only known lower bounds for Online Shadow Tomography are the aforementioned ones for Adaptive Data Analysis, as well as an $\Omega(d^2/\eps^2)$ lower bound that holds when $m \geq \exp(\mathcal{O}(d^2))$ by Aaronson~\cite{Aar20}.

In the easier setting of Offline Shadow Tomography,  additional results are known.    Preceding~\cite{BO24}, Aaronson~\cite{Aar20} showed that $n = \wt{\mathcal{O}}(\log^4(m) \log(d)/\eps^4)$ suffices; the bound of B\u{a}descu--O'Donnell was later obtained via a different method by Bostanci--Bene{ }Watts~\cite{BB24}.)  In the dimension-independent setting, Sinha~\cite{Sin25} gave an efficient algorithm achieving $n = \mathcal{O}(\sqrt{m} \log(m)/\eps^2)$.  Chen--Li--Liu~\cite{chen2024optimal} obtained the optimal $n = \mathcal{O}(\log(m)/\eps^2)$, but only under the assumption $\eps = \mathcal{O}(1/d^{12})$; this assumption was weakened to $\eps = \mathcal{O}(1/d)$ by Pelecanos--Spilecki--Wright~\cite{PSW26}.

Finally, we mention that additional offline Shadow Tomography results are known for special classes of observables~\cite{HKP20,CGY24,KGKB25,HLHP26} (the last of these is online), for special classes of states~\cite{GPS24,CG26}. Interestingly, even when $d = 2^N$ and the observables are specialized to the set of all $N$-qubit Pauli operators, it was not known how to achieve better than quartic dependence in $1/\epsilon$, and~\cite{CGY24} proved that any such algorithm must use highly entangled measurements.

\paragraph{Quantum Efron--Stein decomposition.} The key technical tool in our work, the \emph{excitation decomposition} (see Section~\ref{sec:pickl}), was originally introduced by Pickl~\cite{pickl2011simple} as a simple method for deriving mean-field limits of quantum systems. The method has subsequently been used and extended at length in the quantum mean-field literature, see e.g., ~\cite{knowles2010mean,lewin2015bogoliubov,lewin2015fluctuations,pickl2015derivation}.

As we discuss in Remark~\ref{remark:efronstein}, this decomposition is formally the dual of the \emph{quantum Efron-Stein decomposition} recently developed by the second author and others~\cite{de2025non} in the context of quantum state certification in non-iid settings. To the best of our knowledge, even classically, these tools have not been used in the context of Adaptive Data Analysis.

\section{Technical overview}

In this section we provide a high-level description of the core ideas behind Theorems~\ref{thm:main_logm} and~\ref{thm:main_sqrtm}. For simplicity, throughout this overview we will assume that $\rho = \ketbra{\psi}{\psi}$ by passing to a purification. This turns out to be without loss of generality as our protocol never touches the purification register.

\subsection{Basic terminology}

\paragraph{Sub-normalized post-measurement states.} Given a positive operator-valued measure (POVM) $(M_y)$ implemented with Kraus operators $(\sqrt{M_y})$, and given a state $\sigma$, the post-measurement state upon observing outcome $y$ is given by $\frac{\sqrt{M_y}\sigma\sqrt{M_y}}{\Tr(M_y\sigma)}$. We often reference the \emph{sub-normalized} post-measurement state $\sqrt{M_y}\sigma\sqrt{M_y}$ rather than the post-measurement state, as the trace of the former encodes additional information, namely the probability of observing that measurement outcome. We can naturally extend this notion to a \emph{sequence} of measurements: given POVMs $(M^{(1)}_y),\ldots,(M^{(t)}_y)$ with similarly constructed Kraus operators, if we observe a sequence of outcomes $y_1,\ldots,y_t$, then the resulting sub-normalized state is $\sqrt{M^{(t)}_{y_t}}\cdots \sqrt{M^{(1)}_{y_1}} \sigma \sqrt{M^{(1)}_{y_1}}\cdots \sqrt{M^{(t)}_{y_t}}$, and its trace is the probability of observing $y_1,\ldots,y_t$.

\paragraph{Quantum events.} Given a finite-dimensional Hermitian operator $A\in\mathbb{C}$ with eigendecomposition $A = \sum_\lambda \lambda \Pi_\lambda$, and given some $S\subseteq\mathbb{R}$, we use $\mathbf{1}_{A\in S}$ to denote the projector to the joint span of $\Pi_{\lambda}$ for all $\lambda\in S$. For instance, we write $\mathbf{1}_{|A - \nu \Id| \le \epsilon}$ to denote the projector to the joint span of $\Pi_\lambda$ for which $|\lambda - \nu| \le \epsilon$.

\subsection{Excitation decomposition} 

The main difficulty with measuring multiple observables of a quantum state is that quantum measurement is inherently destructive. Prior works on Shadow Tomography have developed various ways of \emph{gently} measuring the state and bounding the damage incurred in terms of how much information the measurement reveals. For instance, the \emph{gentle measurement lemma}~\cite{winter1999coding,ogawa2002new} bounds the trace distance between the original state $\rho$ and the post-measurement state under a two-outcome POVM $(A,\Id-A)$ in terms of $\Tr(A\rho)$. This bound is pervasive in this literature and in fact tight for worst-case observables $A$.

Our key starting point is the observation that the observables for which the gentle measurement lemma is typically applied in standard approaches to Shadow Tomography come with additional structure. Indeed, in all state-of-the-art works on this problem, the $A$'s in question are \emph{multi-copy} observables derived by a common recipe: starting from some single-copy observable $B$, \emph{lift} it to a multi-copy observable by forming
\begin{equation}
    \overline{B} = \frac{1}{n}\sum^n_{i=1} B_i \qquad \text{for} \qquad B_i \triangleq \Id^{\otimes (i-1)} \otimes B\otimes \Id^{\otimes (n-i)}\,, \label{eq:lift_overview}
\end{equation}
and define a spectral function $A = f(\overline{B})$.

In this work, we develop a new framework for reasoning about damage specifically incurred from measuring such observables, inspired by a state decomposition originally introduced by Pickl~\cite{pickl2011simple} in the context of mean-field limits of quantum systems and which is formally dual to the \emph{quantum Efron--Stein decomposition} recently developed by the second author and others~\cite{de2025non}. Specifically, given any pure state $\ket{\phi}$ on $n$ registers, we can write
\begin{equation}
    \ket{\phi} = \sum_{S\subseteq[n]} \ket{\phi_S}\, \qquad \text{for} \qquad \ket{\phi_S} \triangleq \Bigl(\prod_{i\in S} (\Id - \ketbra{\psi}{\psi})_i\Bigr) \Bigl(\prod_{i\not\in S} \ketbra{\psi}{\psi}_i \Bigr)\ket{\phi}\,.
\end{equation}
This is an orthogonal decomposition, so that $(\norm{\ket{\phi_S}}^2_2)_{S\subseteq[n]}$ is a distribution over subsets $S$. Note that when $\ket{\phi} = \ket{\psi}^{\otimes n}$, this distribution places all of its mass on $S = \emptyset$. Our guiding intuition will be that the more mass this distribution places on \emph{large} subsets, the more damaged $\ket{\phi}$ is relative to the original state $\ket{\psi}^{\otimes n}$. 

More quantitatively, we associate to $\ket{\phi}$ an \emph{energy} given by
\begin{equation}
    \energy[\ketbra{\phi}{\phi}] = \frac{1}{n}\bra{\phi}\mathsf{N}\ket{\phi}\,, \qquad \text{for} \qquad \mathsf{N} \triangleq \sum^n_{i=1} (\Id - \ketbra{\psi}{\psi})_i\,.
\end{equation}
In the language of bosonic systems, $\mathsf{N}$ is the \emph{number operator} and $\energy$ is the \emph{mean photon number per register}. This can also be seen as the quantum generalization of \emph{total influence} in the analysis of Boolean functions (see Remark~\ref{remark:efronstein}).

\subsection{Charging argument} 

The energy starts out at $0$ for $\ket{\phi} = \ket{\psi}^{\otimes n}$ and gives a powerful proxy for tracking damage to the system incurred by a sequence of measurements. 

In a given round of the learning protocol, we need to estimate $\Tr(A^{(t)}\rho)$ for some observable $A^{(t)}$; let $\overline{A}^{(t)}$ denote its $n$-copy lift, defined analogously to Eq.~\eqref{eq:lift_overview}. Let $\tau = \ketbra{\phi}{\phi}$ denote the current sub-normalized state of the system on $n$ registers, so that $\Tr(\tau)$ denotes the probability of the protocol reaching this state. 
Now suppose $\tau$ was ``problematic'' in the sense that it is supported on the quantum event $|A^{(t)} - \Tr(A^{(t)}\rho)\Id| > \epsilon$.
% subspace spanned by eigenvectors of $\overline{A}^{(t)}$ with eigenvalue outside of the interval $[\Tr(A^{(t)}\rho) - \epsilon, \Tr(A^{(t)}\rho) + \epsilon]$. 
This certainly would not have been the case at the beginning of the protocol when $\ket{\phi} = \ket{\psi}^{\otimes n}$, and indicates that the system has incurred some damage from the operations performed in previous rounds of the protocol.

The key observation is that if this is the case, then the likelihood $\Tr(\tau)$ of reaching this particular conditional state can be upper bounded by the energy of the system! Specifically, we show (Lemma~\ref{lem:secondmoment}) that 
\begin{equation}
    \Tr(\tau) \lesssim \frac{1}{\epsilon^2} \energy[\tau]\,.\label{eq:charge_overview}
\end{equation}
Intuitively this means that either $\Tr(\tau)$ is small in which case it's unlikely we will reach this problematic state anyways, or if it is large, then it must mean that by that point, the system has very high energy. As a result, if we could ensure that every operation performed on the system in our protocol does not raise the energy too much in expectation, then the probability that the system ends up in such problematic states is low.

Thus far, we have not used anything about the structure of the operations being performed on the system over the course of the protocol. The measurements used differ markedly between our two results, but the common feature is that they operate over the eigenbasis of the lifted observable $\overline{A}^{(t)}$ in each round. This lifted structure is essential to controlling the energy increase in each round.

\subsection{Controlling energy increase under many observables}

In the case of Theorem~\ref{thm:main_logm}, our protocol follows the general \emph{online learning} template used in prior works on Shadow Tomography~\cite{Aar20,aaronson2019gentle,BO24,BB24}. We attempt to implement a two-player game between a ``student'' and a ``teacher,'' where the student iteratively updates an estimate for the unknown state, and the teacher, equipped with copies of the unknown state, rejects or accepts the student's estimate depending on whether it correctly predicts a given observable value.

The implementation for the student is standard: based on the teacher's feedback, the student runs standard matrix multiplicative weights updates~\cite{ACH+19}. In our work, we implement the teacher as follows. Instead of directly estimating the true value of $\Tr(A^{(t)}\rho)$ in a given round, which naively might require a highly destructive measurement, the teacher merely tries to determine whether the student's current estimate is close by performing the following measurements. Let $\nu$ be the student's estimate for $\Tr(A^{(t)}\rho)$, and consider the two-outcome, $n$-copy measurement $(f(\overline{A}^{(t)}), \Id - f(\overline{A}^{(t)}))$ where $f$ is the logistic function
\begin{equation}
    f(x) = \frac{1}{1 + e^{\lambda(\nu - c\epsilon - x)}}
\end{equation}
for appropriately chosen parameters $\lambda, c$.
By design, the probability $\Tr(f(\overline{A}^{(t)})\rho)$ of the first measurement outcome is close to $1$ (resp. $0$) when the student's estimate is smaller (resp. larger) than the true value of $\Tr(A^{(t)}\rho)$ by some margin. In particular, if the student's estimate is inaccurate, the outcome of this measurement is sufficiently biased that we do not expect the system to be damaged too much. Our core technical step (Corollary~\ref{cor:damage}) is to show that in expectation over this measurement, if the conditional state of the system is $\tau$, then after the measurement the energy increases by at most 
\begin{equation}
    \mathcal{O}\Bigl(\frac{\lambda^2}{n^2}\Bigr)\cdot \Tr(f(\overline{A}^{(t)})\tau)\,. \label{eq:energy_overview1}
\end{equation}
Observe that the sum of $\Tr(f(\overline{A}^{(t)}\tau))$ over all possible observables $A^{(t)}$ and corresponding conditional states $\tau$ over the course of the protocol is equal to the expected number of mistakes that the student makes interacting with the teacher, which by standard regret minimization guarantees scales with $\mathcal{O}(\log(d)/\epsilon^2)$. So for the sum of the expected energy increments in Eq.~\eqref{eq:energy_overview1} to be bounded by $\mathcal{O}(\epsilon^2)$ so that we may invoke Eq.~\eqref{eq:charge_overview}, we need to take 
\begin{equation}
    \lambda = \mathcal{O}(n\epsilon^2/\sqrt{\log(d)})\,.
\end{equation}
The argument is concluded upon noting that the conditional states $\tau$ encountered in the protocol are never fully ``problematic'' in the above sense of being entirely supported on the subspace given by $|\overline{A}^{(t)} - \Tr(A^{(t)}\rho)\Id| > \epsilon$. For such conditional states $\tau$, in addition to the energy term in Eq.~\eqref{eq:energy_overview1}, there is an extra term that scales with $e^{-\Omega(\lambda\epsilon)}$ (Lemma~\ref{lem:selfbound1}) coming from the slope of the logistic function. Summing this over the $O(m)$ total measurements that the teacher performs over the course of the full protocol, and by our choice of $\lambda$, we conclude that this excess term $me^{-\Omega(\lambda\epsilon)}$ is sufficiently small provided that $n = \Omega(\log m\cdot \sqrt{\log d} / \epsilon^3)$.

\subsection{Controlling energy increase under few observables}
In the case of Theorem~\ref{thm:main_sqrtm}, our protocol is much simpler and involves directly measuring the observables queried by the adversary, but corrupted with a small amount of noise. This is similar in spirit to a version of the protocol of Sinha~\cite{Sin25} in the special case of fixed, non-adaptively chosen observables, which was in turn inspired by the Gaussian mechanism from Differential Privacy and Adaptive Data Analysis~\cite{dwork2006our,bassily2021algorithmic}. A small difference is that instead of Gaussian noise, we use compactly supported noise to avoid extra $\log(m)$ factors. Specifically, we consider adding noise drawn from the distribution over $[-\omega,\omega]$, for width parameter $\omega = \Theta(\epsilon)$,
\begin{equation}
    \phi(x) = \omega^{-1/2}\cos(\pi x/2\omega)\cdot \mathbf{1}_{|x|\le \omega}\,.
\end{equation}
Analogous to before, the core technical step (Lemma~\ref{lem:sqrtm_perstep}) is to show that in expectation over such a measurement, if the conditional state of the system is $\tau$, then after the measurement the energy increases by at most
\begin{equation}
    \mathcal{O}\Bigl(\frac{1}{n^2\epsilon^2}\Bigr)\,.
\end{equation}
So after $m$ rounds, the energy increases in expectation to order $\frac{m}{n^2\epsilon^2}$, and by Eq.~\eqref{eq:charge_overview}, we need this to be of order $\epsilon^2$ in order for the probability of reaching a problematic state to be small. Solving for $n$, we find that this can be achieved with $n = \Omega(\sqrt{m}/\epsilon^2)$ copies.

\section{Excitation decomposition}
\label{sec:pickl}

The key tool in our analysis is to quantify post-measurement damage using \emph{excitations}, inspired by the counting method of Pickl~\cite{pickl2011simple}.

Let $\ket{\psi}$ be any purification of the unknown state $\rho$, let $\mathcal{H}$ and $\mathcal{R}$ denote the original Hilbert space and the purification registers respectively, and define projectors
\begin{equation}
    P = \ketbra{\psi}{\psi}\,, \qquad Q = \Id - P\,.
\end{equation}
Given $n$ copies of $\rho$, for every subset $S\subseteq[n]$, define
\begin{equation}
    \Pi_S \triangleq \Bigl(\prod_{i\in S} Q_i\Bigr)\Bigl(\prod_{i\not\in S} P_i\Bigr)\,.
\end{equation}
Any vector $\ket{\phi}\in (\mathcal{H}\otimes \mathcal{R})^{\otimes n}$ admits a unique decomposition
\begin{equation}
    \ket{\phi} = \sum_{S\subseteq[n]} \ket{\phi_S}\,, \qquad \ket{\phi_S} \triangleq \Pi_S \ket{\phi}\,,
\end{equation}
in analogy with the classical \emph{Efron--Stein decomposition} for functions over product spaces. Here, the vectors $\ket{\phi_S}$ are orthogonal, and $\norm{\phi}^2_2 = \sum_S \norm{\phi_S}^2_2$.

% Given $0 \le k \le n$, define \sitan{might not need $\Pi_{\le k}$ and $\Pi_{>k}$ anymore...}
% \begin{equation}
%     \Pi_k\triangleq\sum_{|S|=k}\Pi_S\,, \qquad \Pi_{\le k}\triangleq\sum_{j\le k}\Pi_j\,, \qquad \Pi_{>k}\triangleq\sum_{j>k}\Pi_j
% \end{equation}
Define the \emph{number operator}
\begin{equation}
    \mathsf{N}\triangleq\sum_{i=1}^nQ_i.
    % =\sum_{k=0}^n k\cdot \Pi_k\,.
\end{equation}
We say that $\ket{\phi}$ has \emph{excitation degree at most $k$} if $\ket{\phi_S} = 0$ for all $|S| > k$.

We can now define the central quantity that allows us to reason about how much damage a state has incurred from measurements.

\begin{definition}[Energy]
    The \emph{energy} of a sub-normalized mixed state $\tau$ is defined by
    \begin{equation}
        \energy[\tau] \triangleq \frac{1}{n}\Tr(\mathsf{N}\tau)\,.
    \end{equation}
\end{definition}

\begin{remark}\label{remark:efronstein}
    There is a dual interpretation of the excitation decomposition in terms of the \emph{quantum Efron--Stein decomposition} developed in~\cite{de2025non}. Given an operator $B$ over $\mathcal{H}^{\otimes n}$, one can define operators 
    \begin{equation}
        \mathcal{E}_i(B) = \Tr_i(\rho_i B)\otimes \Id_i \qquad \text{and} \qquad \mathcal{D}_i(B) = \Id - \mathcal{E}_i
    \end{equation} and projectors
    \begin{equation}
        \mathcal{P}_S \triangleq \Bigl(\prod_{i\in S} \mathcal{D}_i\Bigr)\Bigl(\prod_{i\in S}\mathcal{E}_i\Bigr)\,,
    \end{equation}
    and the \emph{quantum Efron--Stein decomposition} of $B$ is $B = \sum_S \mathcal{P}_S B$. The connection to the excitation decomposition of a state $\ket{\phi}$ is that if $\ket{\phi}$ is the sub-normalized conditional state given by applying a sequence of measurements where the product of the associated Kraus operators is $B$, then $B_S$ and $\ket{\phi_S}$ are related via
    \begin{equation}
        (B_S \otimes \Id_{\mathcal{R}})\ket{\psi}^{\otimes n} = \ket{\phi_S}
    \end{equation}
    One can thus think of the excitation decomposition and the Efron--Stein decompositions as Schr\"odinger / Heisenberg duals of each other. Furthermore, the \emph{Efron--Stein degree operator} $\mathsf{N}_{\rm ES} = \sum^n_{i=1} \mathcal{D}_i$ is the dual of the number operator in the excitation decomposition, and if $B$ were a classical operator corresponding to a function over $[d]^n$, then the corresponding energy $\energy[\ketbra{\phi}{\phi}]$ would be the \emph{total influence} of the function in the sense of classical functions over product spaces.
\end{remark}

\noindent If $\ket{\phi}$ is normalized, then $(\norm{\Pi_S\phi}^2_2)_{S\subseteq[n]}$ specifies the probability mass function for a random variable $S\subseteq[n]$; intuitively, if $\ket{\phi}$ is a post-measurement state, then the more mass is concentrated at smaller subsets $S$, the less ``damaged'' it is. 

\subsection{Lifted observables and energy}
\label{sec:lifted}

\begin{definition}
    Given a number of copies $n$ and a single-copy Hermitian observable $0 \preceq A\preceq \Id$, define the \emph{lifted observable} $\lift{A}$ by
    \begin{equation}
        \lift{A} = \frac{1}{n}\sum^n_{j=1} A_j\,, \qquad \text{for} \qquad A_j \triangleq \Id^{\otimes (j-1)}\otimes A\otimes \Id^{\otimes (n-j)}\,.
    \end{equation}
\end{definition}

% We now show that for states whose mass is mostly concentrated in low-excitation subspaces, the value of lifted observables is close to the value of the original single-copy observable.

% \begin{lemma}[Low excitation implies accurate estimate]\label{lem:lowexcite}
%     There are absolute constants $c_1, c_2, C > 0$ such that the following holds for any $0 < \xi \le 1$. If
%     \begin{equation}
%         n\ge c_1 / \xi^2 \qquad \text{and} \qquad k \le c_2\xi^2 n\,, \label{eq:assume_bounds_nk}
%     \end{equation}
%     then for any vector $\ket{\phi}$ satisfying $\norm{\phi}_2 \le 1$, and for the true mean $\mu = \Tr(A\rho)$ and lifted observable $\lift{A}$ associated with observable $0 \preceq A\preceq \Id$, we have
%     \begin{equation}
%         \norm{\mathbf{1}_{|\lift{A} - \mu \Id| > \xi} \Pi_{\le k} \ket{\phi}}^2_2 \le \exp(-Cn\xi^2)\,.
%     \end{equation}
% \end{lemma}

\noindent For convenience, let $X \triangleq (A - \mu \Id)\otimes \Id_{\mathcal{R}}$, and let $\lift{X} = \lift{A} - \mu \Id$. Note that $\norm{X}_{\sf op} \le 1$ and
\begin{equation}
    PXP = \bra{\psi}X\ket{\psi}P = (\Tr(A\rho) - \mu)P = 0\,.
\end{equation}
Because $PXP = 0$, we can decompose $X$ as
\begin{equation}
    X = X_+ + X_- + X_0,\, \qquad X_+ \triangleq QXP\,, \quad \text{where} \quad X_- \triangleq PXQ\,, \quad X_0 \triangleq QXQ.
\end{equation}
Note that $X^\dagger_+ = X_-$ and $X^\dagger_0 = X_0$.
Informally, the left action of $X_+$ corresponds to adding an excitation, that of $X_-$ corresponds to removing an excitation, and that of $X_0$ corresponds to keeping an excitation. 

\begin{lemma}\label{lem:secondmoment}
    Let $\xi > 0$ satisfy $n \xi^2 \geq 18$.
    Let $\tau \in (\calH \otimes \calR)^{\otimes n}$ be a sub-normalized mixed state.
    Consider the quantum event $|\lift{A} - \Tr(A\rho)\Id| \ge \xi$, and let $B = \mathbf{1}_{|\lift{A} - \Tr(A\rho)\Id|\ge \xi}$ be the projector onto the corresponding subspace.
    Then
    \begin{equation*}
        \Tr(B \cdot \tau)
        \leq \frac{9}{\xi^2} \energy[\tau]
        + \frac{1}{2} \Tr(\tau).
    \end{equation*}
    As a consequence, if $\tau$ lies in the subspace corresponding to $B$, then
    \begin{equation}
        \Tr(\tau) \le \frac{18}{\xi^2} \energy[\tau]\,.
    \end{equation}
\end{lemma}

\begin{proof}
    Let $B = \mathbf{1}_{|\lift{A} - \Tr(A\rho)\Id|\ge \xi}$ be the projector to the subspace given by this event.
    By definition, we have
    \begin{equation*}
        B
        \preceq \frac{1}{\xi^2} 
        \cdot (\lift{A} - \Tr(A\rho)\Id)^2
        = \frac{1}{\xi^2} 
        \cdot \lift{X}^2.
    \end{equation*}
    Applying Lemma~\ref{lem:NI} below, we have that
    \begin{equation*}
        \frac{1}{\xi^2} 
        \cdot \lift{X}^2
        \preceq \frac{1}{\xi^2} 
        \cdot \frac{9}{n}(\mathsf{N} + \Id).
    \end{equation*}
    Thus,
    \begin{equation*}
        \Tr(B \cdot \tau)
        \leq \frac{9}{n \xi^2} \Tr((\mathsf{N} + \Id) \tau)
        = \frac{9}{\xi^2} \energy[\tau]
        + \frac{9}{n \xi^2} \Tr(\tau)
        \leq \frac{9}{\xi^2} \energy[\tau]
        + \frac{1}{2} \Tr(\tau),
    \end{equation*}
    where in the last step we used our assumption that $n \xi^2 \geq 18$.
    The last part of the lemma follows from the fact that $\Tr(\tau) = \Tr(B\cdot \tau)$ and from rearranging.
\end{proof}

\begin{lemma}\label{lem:NI}
    $\lift{X}^2 \preceq \frac{9}{n}(\mathsf{N} + \Id)$.
\end{lemma}

\begin{proof}
    Let $\ket{u}$ be a test vector, and write $\ket{u} = \sum_S \ket{u_S}$ for its excitation decomposition. Let $S\subseteq [n]$ be a subset of size $|S| = D$, and consider the vector $\ket{u_S}$.
    Consider the operator $X_i$ in the definition of $\lift{X}$.
    If $i \notin S$, then $P_i \cdot \ket{u_S} = \ket{u_S}$, and so $X_i \cdot \ket{u_S} = X_+ \cdot \ket{u_S}$. In this case, the action of $X_i$ is to excite the $i$-th slot of $\ket{u_S}$, making $X_i \cdot \ket{u_S}$ have excitation degree at most $D + 1$.
    On the other hand, if $i\in S$, then $Q_i \cdot \ket{u_S} = \ket{u_S}$, and so we have the decomposition $X_i\cdot \ket{u_S} = (X_-)_i\cdot \ket{u_S} + (X_0)_i\cdot \ket{u_S}$ into a vector with excitation degree at most $D -1$, and one with excitation degree at most $D$.
    
    More precisely, we find that the excitation decomposition of $\lift{X} \ket{u}$ satisfies, for every $T\subseteq[n]$,
    \begin{equation}
        \Pi_T\lift{X}\ket{u}=\frac{1}{n}\underbrace{\Bigl(\sum_{i\in T}(X_+)_i\cdot \ket{u_{T\setminus\{i\}}}\Bigr)}_{\ket{A_T}}+\frac{1}{n}\underbrace{\Bigl(\sum_{i\notin T}(X_-)_i\cdot \ket{u_{T\cup\{i\}}}\Bigr)}_{\ket{B_T}}+\frac{1}{n}\underbrace{\Bigl(\sum_{i\in T}(X_0)_i\cdot \ket{u_T}\Bigr)}_{\ket{C_T}}\,. \label{eq:ABC_quantum}
    \end{equation}
    Thus, if we write $\ket{A} = (1/n) \cdot \sum_T \ket{A_T}$, and similarly for $\ket{B}$ and $\ket{C}$, we have
    $\lift{X} \ket{u} = \ket{A} + \ket{B} + \ket{C}$.
    As a result,
    \begin{align}
        \bra{u} \lift{X}^2\ket{u}
        = \norm{\lift{X} \ket{u}}^2
        = \norm{\ket{A} + \ket{B} + \ket{C}}^2
        &\leq 3 \cdot (\norm{\ket{A}}^2 + \norm{\ket{B}}^2 + \norm{\ket{C}}^2)\nonumber\\
        &=\frac{3}{n^2} \cdot \Big(\sum_T \norm{\ket{A_T}}^2+\sum_T \norm{\ket{B_T}}^2+\sum_T \norm{\ket{C_T}}^2\Big),\label{eq:Deltaunorm_quantum}
    \end{align}
    where in the last step we use the fact that the $\ket{A_T}$'s are orthogonal to each other, as are the $\ket{B_T}$'s and $\ket{C_T}$'s.
    Let us bound these terms separately. First, for a fixed excitation degree $D'$, we have
    \begin{align*}
    \sum_{|T|=D'}\norm{\ket{A_T}}^2&\le\sum_{|T|=D'}D'\cdot \sum_{i\in T}\norm{\ket{u_{T\setminus\{i\}}}}^2=D'(n-D'+1)\cdot \sum_{|S|=D'-1}\norm{\ket{u_S}}^2, \\
    \sum_{|T|=D'}\norm{\ket{B_T}}^2&\le\sum_{|T|=D'}(n-D')\cdot \sum_{i\notin T}\norm{\ket{u_{T\cup\{i\}}}}^2=(n-D')(D'+1)\cdot \sum_{|S|=D'+1}\norm{\ket{u_S}}^2, \\
    \sum_{|T|=D'}\norm{\ket{C_T}}^2&\le\sum_{|T|=D'}D'\cdot \sum_{i\in T}\norm{\ket{u_T}}^2=(D')^2\cdot \sum_{|T|=D'}\norm{\ket{u_T}}^2.
    \end{align*}
    In the first step of all three lines, we have used Cauchy-Schwarz and the fact that $X_+, X_-, X_0$ have operator norm at most $1$.
    Summing these inequalities, we have
    \begin{align}
        \sum^n_{D'=1} \sum_{|T| = D'}\norm{\ket{A_T}}^2 &\le n\cdot \sum^n_{D'=1} D' \sum_{|S| = D'-1} \norm{\ket{u_S}}^2 \leq n\cdot \bra{u}(\mathsf{N} + \Id)\ket{u},\\
        \sum^{n-1}_{D'=0} \sum_{|T| = D'}\norm{\ket{B_T}}^2 &\le n\cdot \sum^{n-1}_{D'=0} (D'+1)\sum_{|S| = D'+1} \norm{\ket{u_S}}^2 \le n\cdot \bra{u}\mathsf{N}\ket{u},\\
        \sum^n_{D'=0} \sum_{|T| = D'}\norm{\ket{C_T}}^2 &\le n\cdot \sum^n_{D'=0} D'\sum_{|T| = D'} \norm{\ket{u_T}}^2 = n\cdot \bra{u}\mathsf{N}\ket{u} \,,
    \end{align}
    and substituting these into Eq.~\eqref{eq:Deltaunorm_quantum} yields the claimed bound.
\end{proof}

% \sitan{Lemma~\ref{lem:NI} is very slick but I wonder if we should still keep the fact that low excitation degree implies exponential concentration, just to give the reader intuition for why excitation degree corresponds to low damage?}

\section{Logarithmic-in-\texorpdfstring{$m$}{m} rate}
\label{sec:logm}

In this section we prove the upper bound in Theorem~\ref{thm:main_logm}, restated here for convenience:

\logm*

\subsection{Description of protocol}
\label{sec:protocol}

% \angelos{changed $\gap$'s to $\epsilon/4$ and adjusted a few constants}
% gap: gap parameter for threshold search
% \begin{equation}
%     \gap = \epsilon/4\,, \qquad 
% \end{equation}

\begin{definition}[Low and high tests]
    Given observable $0\preceq A\preceq \Id$ and threshold $\theta$, we say that
    \begin{align}
        ``(A, \theta) \ \mathrm{is \ low}" \ & \ \text{if} \ \ \Tr(A\rho) \le \theta - \frac{\epsilon}{4}, \\
        ``(A, \theta) \ \mathrm{is \ high}" \ & \ \text{if} \ \ \Tr(A\rho) > \theta\,.
    \end{align}
    The protocol does not know whether a given $(A,\theta)$ is low or high, and this terminology is only used for the analysis.
\end{definition}

\begin{definition}[Soft binary measurements]
    For a lifted observable $\lift{A}$ and threshold $\theta$, define the \emph{soft-thresholded observable} $f_{\theta,\lambda}(\lift{A})$, where
    \begin{equation}
        f_{\theta, \lambda}(x) = \frac{1}{1 + e^{\lambda(\theta - \epsilon/8 - x)}}\,.
    \end{equation}
    The choice of $\theta - \epsilon/8 - x$ ensures that the logistic function is centered halfway between the low and high thresholds of $\theta - \epsilon/4$ and $\theta$. We refer to performing the two-outcome measurement corresponding to $(f_{\theta,\lambda}(\lift{A}), \Id - f_{\theta,\lambda}(\lift{A}))$ as performing a \emph{soft binary measurement} corresponding to $(A,\theta)$, with the two outcomes referred to as \emph{flag} and \emph{pass} respectively. More precisely, we implement this measurement using Kraus operators
    \begin{equation*}
        \Big(\sqrt{f_{\theta,\lambda}(\lift{A})}, \sqrt{\Id - f_{\theta,\lambda}(\lift{A})}\Big).
    \end{equation*}
\end{definition}

\noindent We will analyze the procedure specified in Algorithm~\ref{alg:shadow-logm}.

\begin{algorithm}[H]
\DontPrintSemicolon
\caption{\textsc{OnlineShadowTomography1}}
\label{alg:shadow-logm}
\SetKw{KwAnd}{and}
\SetKw{KwOr}{or}
\SetKwInOut{Params}{Parameters}
\SetKwInOut{StateVar}{State}
\SetKwComment{Cmt}{$\triangleright$~}{}
\Params{number of rounds $m$;\ \ target accuracy $\epsilon$;\ \ mistake budget $K = \Theta(\log(d)/\epsilon^2)$;\ \ logistic sharpness $\lambda = \Theta(\epsilon n/\sqrt{K})$}
% \StateVar{Hamiltonian $H \gets 0$ with Gibbs state $\rho_H \triangleq e^{-H}/Z_H$;\ \ mistake counter $n_{\sf mstk} \gets 0$;\ \ flag $\mathsf{failed} \gets \textsc{false}$.}
$H\gets 0$;\quad $n_{\sf mstk} \gets 0$;\quad $\mathsf{failed}\gets \textsc{false}$\;
\For{$t = 1, 2, \dots, m$}{
  Receive observable $A^{(t)}$, chosen by the adversary as a function of the transcript $\calT_{t-1}$.\;
  \eIf{$\mathsf{failed}$}{
    Output $\widehat{\mu}^{(t)} \gets \tfrac12$\Cmt*[r]{algorithm has already failed}
  }{
    \Repeat{$L = \bot$ \KwOr $\mathsf{failed}$}{
      $\rho_H \gets e^{-H}/\Tr(e^{-H})$\;
      $\nu \gets \Tr\!\big(A^{(t)}\rho_H\big)$;\quad $L \gets \bot$\Cmt*[r]{$\bot$: no test has flagged this time}
      \Cmt{high test on $A^{(t)}$; skipped when $\nu$ is in the top band $(1-3\epsilon/4,\,1]$}
      \If{$\nu \le 1 - 3\epsilon/4$}{
        Perform the soft binary measurement corresponding to $\big(A^{(t)},\ \nu + 3\epsilon/4\big)$.\;
        \lIf{it flags}{$L \gets -A^{(t)}$}
      }
      \Cmt{high test on $\Id - A^{(t)}$; skipped when $\nu$ is in the bottom band $[0,\,3\epsilon/4)$, or when the first test already flagged}
      \If{$L = \bot$ \KwAnd $\nu \ge 3\epsilon/4$}{
        Perform the soft binary measurement corresponding to $\big(\Id - A^{(t)},\ 1 - \nu + 3\epsilon/4\big)$.\;
        \lIf{it flags}{$L \gets A^{(t)}$}
      }
      \Cmt{if some test has flagged, reweight and count the mistake}
      \If{$L \ne \bot$}{
        $H \gets H + \frac{\epsilon}{8}L$;\quad $n_{\sf mstk} \gets n_{\sf mstk} + 1$\;
        \lIf{$n_{\sf mstk} = K$}{$\mathsf{failed} \gets \textsc{true}$}
      }
    }
    \eIf{$\mathsf{failed}$}{Output $\widehat{\mu}^{(t)} \gets 1/2$\Cmt*[r]{mistake budget exhausted}}{Output $\widehat{\mu}^{(t)} \gets \nu$}
  }
}
\end{algorithm}

% \noindent Here $\bot$ denotes the absence of an passed test within a pass of the inner \textbf{repeat} loop, and ``soft binary measurement corresponding to $(A,\theta)$'' is the two-outcome measurement $\big(f_{\theta,\lambda}(\lift{A}),\, \Id - f_{\theta,\lambda}(\lift{A})\big)$ defined above. 
% We refer to the above as round $t$ for a single batch of $n$ copies.

% To amplify the success probability of this protocol, we will run the above procedure in parallel across $O(\log 1/\delta)$ batches of copies in parallel. That is, after the adversary chooses an observable $A^{(t)}$, the learner runs $O(\log 1/\delta)$ parallel instantiations of the above procedure and aggregates their outputs by taking the median.

\begin{definition}
    % For a single batch of $n$ copies, 
    We say that round $t$ is \emph{bad} if at some point therein, the protocol flags for a low test $(A,\theta)$ or passes for a high test $(A,\theta)$. The learner does not know if a given round is bad, and this terminology is only used for the analysis.
\end{definition}

\noindent The following is a consequence of standard regret guarantees for matrix multiplicative weights:

\begin{lemma}\label{lem:mmw}
    % For a single batch of $n$ copies, 
    If over the course of $m$ rounds of interaction there is never a bad round, then every estimate the protocol outputs satisfies $|\widehat{\mu}^{(t)} - \Tr(A^{(t)}\rho)| < \epsilon$.
\end{lemma}

\begin{proof}
    Fix a round $t$ and a pass through the inner loop, and let $\rho_H$ and $\nu = \Tr(A^{(t)}\rho_H)$ denote the hypothesis state and the value computed at the start of that pass.
    
    We first verify that in an iteration of the inner loop in which all tests pass, the output $\widehat{\mu}^{(t)} = \nu$ is accurate. Since the round is not bad and all tests passed, no test performed in this pass is high. If the first test was performed, then $\Tr(A^{(t)}\rho) \le \nu + 3\epsilon/4$; if it was skipped, then $\nu > 1 - 3\epsilon/4$ and the same inequality holds trivially since $\Tr(A^{(t)}\rho)\le 1$. Likewise, if the second test was performed, then $\Tr(A^{(t)}\rho)\ge \nu - 3\epsilon/4$; if it was skipped, then $\nu < 3\epsilon/4$ and the same inequality holds trivially since $\Tr(A^{(t)}\rho)\ge 0$. Hence $|\widehat{\mu}^{(t)} - \Tr(A^{(t)}\rho)| \le 3\epsilon/4 < \epsilon$, as desired.

    Next, suppose instead that in an iteration of the inner loop, some test flags, so that the learner updates $H \gets H + \frac{\epsilon}{8} L$. Since the round is not bad, the flagged test is not low. In the case of either test, the flag identifies a direction along which the hypothesis is inaccurate by a non-negligible margin:
    \begin{equation}
        \Tr(L(\rho_H - \rho)) > \epsilon/2\,. \label{eq:mmw_margin}
    \end{equation}
    
    Finally, we argue that there cannot be so many flagged steps. The flagged steps constitute a run of matrix multiplicative weights: the learner plays the Gibbs states $\rho_H \propto \exp(-\frac{\epsilon}{8}\sum L)$ against adaptively chosen loss matrices $L = \pm A^{(t)}$, which satisfy $\norm{L}_{\sf op}\le 1$. Suppose $M$ flags occur over the course of the interaction, and let $\rho_j$ and $L_j$ denote the hypothesis and the loss matrix at the $j$-th flag. The standard regret bound for matrix multiplicative weights with step size $\epsilon/8$ (see e.g.~\cite[Corollary 3]{kale2007efficient}) gives
    \begin{equation}
        \sum^M_{j=1}\Tr(L_j(\rho_j - \rho)) \le \frac{\epsilon}{8}\sum^M_{j=1} \abs{\Tr(L_j\rho_j)} + \frac{8\log d}{\epsilon} \le \frac{\epsilon M}{8} + \frac{8\log d}{\epsilon}\,,
    \end{equation}
    whereas Eq.~\eqref{eq:mmw_margin} lower bounds the left-hand side by $\epsilon M/2$. Rearranging, $M \lesssim \log(d)/\epsilon^2$. So provided the constant in $K = \Theta(\log(d)/\epsilon^2)$ is large enough, the counter $n_{\sf mstk}$ never reaches $K$, and each inner loop terminates and outputs an accurate estimate.
\end{proof}

\subsection{Bad outcomes are unlikely}
\label{sec:badunlikely}

% The general strategy will be to argue that as long as the state is undamaged (in the sense that its high-excitation component is small), then nothing bad will happen, and furthermore the damage accumulates benignly. The former is a direct consequence of Lemma~\ref{lem:lowexcite}:

Consider the binary tree of possible complete internal histories of the protocol, with each node labeled by the sub-normalized conditional state at that point in the protocol, and where each leaf node corresponds to either:
\begin{itemize}
    \item The first bad outcome that occurs in that branch of the tree, or
    \item the final outcome in the protocol, where all outcomes along that branch were non-bad.
\end{itemize}
Because we clip the protocol after $n_{\sf mstk} = K$ mistakes corresponding to flags, and because in each of the $m$ rounds we perform at most two soft binary measurements prior to passing, the depth of the tree is at most $2(m+K)$.

Let $\mathcal{D}$ denote the edges of the tree corresponding to the first bad outcome occurring along that path. For each such edge $e\in\mathcal{D}$, let $H_e \in \{F_e, \Id - F_e\}$ denote the soft binary measurement operator corresponding to the bad outcome, and let $\sigma_e$ (resp. $\tau_e$) denote the sub-normalized post-measurement state right before (resp. after) the bad outcome. 

By definition, $\tau_e = \sqrt{H_e}\sigma_e\sqrt{H_e}$. We wish to bound 
\begin{equation}
    \sum_{e\in\mathcal{D}} \Tr(\tau_e)\,, \label{eq:sumbad}
\end{equation}
as this corresponds to the probability over the course of the entire protocol that a bad outcome happens.

The following bounds the probability of reaching a particular bad leaf in terms of the energy of the corresponding sub-normalized conditional state:

\begin{lemma}\label{lem:selfbound1}
    Suppose that at some point in the protocol {\sc OnlineShadowTomography1}, the adversary queries observable $A$, and after performing a soft binary measurement with threshold $\theta$, we get a bad outcome corresponding to measurement operator $H$. If $\tau$ denotes the sub-normalized state in the tree corresponding to this outcome, and $\sigma$ denotes the one prior to measurement, and if $n\epsilon^2 > 4608$, then
    \begin{equation}
        \Tr(\tau) \le 2e^{-\lambda\epsilon/16}\Tr(\sigma) + \frac{4608}{\epsilon^2}\cdot \energy[\tau].
    \end{equation}
\end{lemma}

\begin{proof}
    Write $\mu = \Tr(A\rho)$ and recall the notation $\lift{X} = \lift{A} - \mu \Id$.
    Let us first consider the case when $(A,\theta)$ is low, which we recall means that $\mu \leq \theta - \eps/4$.
    In this case, the measurement operator is $H = f_{\theta,\lambda}(\lift{A})$. Define $\Pi = \mathbf{1}_{\lift{X} > \epsilon/16}$ and let $\overline{\Pi}$ denote its orthogonal complement. Because $H$ and $\Pi$ commute,
    \begin{equation}
        \Tr(\tau) = \Tr(H\sigma) = \Tr(H\Pi\sigma\Pi) + \Tr(H\overline{\Pi}\sigma\overline{\Pi}) = \Tr(\Pi\tau\Pi) + \Tr(H\overline{\Pi}\sigma\overline{\Pi}). \label{eq:taubound}
    \end{equation}
    Over the subspace $\overline{\Pi}$ we have
    \begin{equation}
        \overline{A} \le \mu + \epsilon/16 \le \theta - 3\epsilon/16\,,
    \end{equation}
    where the second step used $\mu \leq \theta - \eps/4$ because $(A, \theta)$ is low.
    Recalling the definition of $f_{\theta,\lambda}$, we have that for any $x \leq \theta - 3\eps/16$,
    \begin{equation*}
        f_{\theta,\lambda}(x)
        = \frac{1}{1 + e^{\lambda(\theta - \epsilon/8 - x)}}
        \leq \frac{1}{1 + e^{\lambda(\theta - \epsilon/8 - (\theta - 3\eps/16))}}
        = \frac{1}{1 + e^{\lambda \eps/16}}
        \leq \frac{1}{e^{\lambda \eps/16}}.
    \end{equation*}
    We therefore get
    \begin{equation*}
    \overline{\Pi} H \overline{\Pi} = \overline{\Pi} f_{\theta,\lambda}(\lift{A})\overline{\Pi}
    % \preceq \|H\|_{\mathrm{op}}\cdot \overline{\Pi}
    % = \|f_{\theta,\lambda}(\lift{A})\|_{\mathrm{op}}\cdot \overline{\Pi}
    \preceq e^{-\lambda\epsilon/16}\overline{\Pi}\,.    
    \end{equation*} 
    This upper bounds the second term of Eq.~\eqref{eq:taubound} by~$e^{-\lambda \epsilon/16} \Tr(\sigma)$.
    As for the first term, set $\xi = \epsilon/16$, and note that $n \xi^2 = n \epsilon^2 /16^2 > 4608 / 16^2 = 18$.
    In addition, set $B = \mathbf{1}_{|\lift{X}| > \epsilon/16}$, and note that $\Pi \preceq B$.
    Then by Lemma~\ref{lem:secondmoment},
    \begin{equation*}
        \Tr(\Pi \tau \Pi) \leq \Tr(B \tau) \le \frac{9}{(\epsilon/16)^2}\energy[\tau] + \frac{1}{2} \Tr(\tau)
        = \frac{2304}{\epsilon^2} \energy[\tau] + \frac{1}{2} \Tr(\tau).
    \end{equation*}
    Substituting these into Eq.~\eqref{eq:taubound} gives us
    \begin{equation*}
        \Tr(\tau) \leq e^{-\lambda \epsilon/16} \Tr(\sigma) + \frac{2304}{\epsilon^2} \energy[\tau] + \frac{1}{2} \Tr(\tau),
    \end{equation*}
    and rearranging yields the claimed bound. The case of high $(A,\theta)$ follows verbatim.
\end{proof}

\noindent The following quantifies the amount by which $\energy[\tau]$ changes in expectation when going from $\tau$ to its expected post-measurement state.

\begin{corollary}\label{cor:damage}
    Let $F = f_{b,\lambda}(\lift{A})$. Let $\tau$ be an arbitrary mixed state, and let
    \begin{equation}
        \luders_F(\tau)\triangleq \sqrt{F}\tau\sqrt{F} + \sqrt{\Id - F}\tau\sqrt{\Id - F}
    \end{equation}
    denote the expected post-measurement state. Then for $0 < \lambda \le n$,
    \begin{equation}
        \energy[\luders_F(\tau)] - \energy[\tau] \le \frac{2\lambda^2}{n^2}\Tr(F\tau)\,. \label{eq:onestep}
        % \luders^\dagger_F(\mathsf{N}) - \mathsf{N}\preceq 2\eta^2 nF\,.
    \end{equation}
\end{corollary}

% By Corollary~\ref{cor:heisenberg_damage}, every soft binary measurement satisfies \sitan{TODO: update $\luders^\dagger_F$ notation}
% \begin{equation}
%     \luders^\dagger_F(\mathsf N) - \mathsf{N} \preceq \frac{2\lambda^2}{n} F\,. \label{eq:onestep}
% \end{equation}

\noindent We defer the proof of this to Section~\ref{sec:perstep} and proceed to formally present the argument outlined above.

\begin{lemma}\label{lem:badbound}
    % Over a single batch of $n$ copies, 
    The probability over the course of the entire protocol that a bad outcome happens is at most
    \begin{equation}
        \frac{9216\lambda^2 K}{n^2\epsilon^2} + 4(m+K) e^{-\lambda\epsilon/16}\,.
    \end{equation}
\end{lemma}

\begin{proof}
    We first bound the total energy across all leaves of the tree via a telescoping calculation.
    Telescoping and noting that $\energy[\tau_{\rm root}] = 0$ for the root of the tree, we have
    \begin{equation}
        \sum_{\mathrm{leaves} \ \tau} \energy[\tau] = \sum_{\tau} \energy[\tau] - \energy[\tau_{\rm root}] = \sum_{\mathrm{inner} \ \sigma} \big(\energy[\luders_{F_\sigma}(\sigma)] - \energy[\sigma]\big) \le \frac{2\lambda^2}{n^2} \sum_{\mathrm{inner} \ \sigma}  \Tr(F_{\sigma}\sigma)\,,
    \end{equation}
    where $F_\sigma$ denotes the soft binary measurement operator used at the inner node $\sigma$ (corresponding to the flag outcome), and where in the last step we used Corollary~\ref{cor:damage}. But note that $\sum_{\mathrm{inner} \ \tau} \Tr(F_\sigma\sigma)$ is the expected number of flags over the course of the protocol, and this is at most the total number of flags which is $K$. We conclude that
    \begin{equation}
        \sum_{e\in\mathcal{D}} \energy[\tau_e] \le \sum_{\mathrm{leaves} \ \tau} \energy[\tau] \le \frac{2\lambda^2 K}{n^2}\,.\label{eq:highexcite_part}
   \end{equation}
    The left-hand side can be lower-bounded by Lemma~\ref{lem:selfbound1}, and upon rearranging we get
    \begin{equation}
        \sum_{e\in\mathcal{D}} (\Tr(\tau_e)- 2 e^{-\lambda\epsilon/16}\Tr(\sigma_e)) \le \frac{9216\lambda^2 K}{n^2\epsilon^2}\,.
    \end{equation}
    The claimed bound on $\sum_{e \in \mathcal{D}}\Tr(\tau_e)$ follows as $\sum_\sigma \Tr(\sigma) \le 2(m+K)$.
\end{proof}

\noindent We can now conclude the proof of the main result:

\begin{proof}[Proof of Theorem~\ref{thm:main_logm}]
    If $n = \Omega(\sqrt{K}\log(m + K)/\epsilon^2)$ with sufficiently large leading constant and  $\lambda = \Theta(\epsilon n/\sqrt{K})$ with sufficiently small leading constant, the quantity $e^{-\lambda\epsilon/16}$ in Lemma~\ref{lem:badbound} is at most $c/(m+K)$ for constant $c > 0$ that can be made arbitrarily small, and the term $\frac{\lambda^2 K}{n^2\epsilon^2}$ in Lemma~\ref{lem:badbound} can also be upper bounded by an arbitrarily small constant, so by Lemma~\ref{lem:badbound}, the probability over the course of the entire protocol that a bad outcome happens is at most an arbitrarily small constant. The proof is complete by Lemma~\ref{lem:mmw}.
\end{proof}

% {\color{gray}[skip this for now and use the simpler potential argument for 0.1 failure probability] For a single batch of $n$ copies, suppose at some point in the protocol the total number of mistakes is $n_{\sf mstk}$, and define the operator
% \begin{equation}
%     W(n_{\sf mstk}) = \mathsf{N} + \Gamma(K - n_{\sf mstk})\Id\,.
% \end{equation}
% for $\Gamma = 2\eta^2 n$.
% Then by Eq.~\eqref{eq:onestep},
% \begin{align}
%     \sqrt{\Id - F} W(n_{\sf mstk}) \sqrt{\Id - F} + \sqrt{F} W(n_{\sf mstk} + 1) \sqrt{F} &= \luders^\dagger_F(\mathsf{N}) + \Gamma(K - n_{\sf mstk})\Id - \Gamma F \\
%     &\preceq W(n_{\sf mstk})\,.
% \end{align}}

\subsection{Per-step damage control: Proof of Corollary~\ref{cor:damage}}
\label{sec:perstep}

We will prove Corollary~\ref{cor:damage} in the Heisenberg picture by controlling the extent to which the number operator $\mathsf{N}$ gets damaged by the adjoint $\luders^\dagger_F$ of the expected post-measurement channel for a soft binary measurement. Our goal is to show that
\begin{equation}
    \luders_F^\dagger(\mathsf{N}) - \mathsf{N} \preceq \frac{2\lambda^2}{n}F\,.
\end{equation}
To understand this, note that in the Schrödinger picture, this is equivalent to the statement that for any state $\sigma$,
\begin{equation*}
    \energy[\luders_F(\sigma)] \leq \energy[\sigma] + \frac{2\lambda^2}{n^2} \Tr(F \cdot \sigma).
\end{equation*}
Thus, the $F$ measurement can only noticeably increase the energy of the state $\sigma$ if there is a decent chance that it produces outcome $F$ when performed on $\sigma$.
This can be viewed as an analogue of the standard \emph{information/disturbance tradeoff} in quantum information.

Define the matrix $U = \sqrt{\Id - F} + i\sqrt{F}$.
This is a unitary, since 
\begin{equation*}
    U U^{\dagger}
    = (\sqrt{\Id - F} + i\sqrt{F}) (\sqrt{\Id - F} - i\sqrt{F})
    = (\Id- F) + F + i \sqrt{F} \sqrt{\Id - F} - i \sqrt{\Id - F}\sqrt{F}
    = \Id,
\end{equation*}
and similarly for $U^{\dagger} U$.
Then for any operator $X$, the post-measurement channel can be written in terms of $U$ as follows:
\begin{equation}
    \label{eq:logm-channel-time-evolution}
    \luders^\dagger_F(X) = \frac{1}{2}(U^\dagger X U + U X U^\dagger)\,.
\end{equation}
This is because
\begin{align*}
    &\frac{1}{2}(U^\dagger X U + U X U^\dagger)\\
    ={}& \frac{1}{2}(\sqrt{\Id - F} + i\sqrt{F})^\dagger\cdot  X \cdot (\sqrt{\Id - F} + i\sqrt{F}) + \frac{1}{2} (\sqrt{\Id - F} + i\sqrt{F})\cdot X \cdot (\sqrt{\Id - F} + i\sqrt{F})^\dagger\\
    ={}& \frac{1}{2}(\sqrt{\Id - F} - i\sqrt{F})\cdot  X \cdot (\sqrt{\Id - F} + i\sqrt{F}) + \frac{1}{2} (\sqrt{\Id - F} + i\sqrt{F})\cdot X \cdot (\sqrt{\Id - F} - i\sqrt{F})\\
    ={}& \sqrt{\Id - F} \cdot X \cdot \sqrt{\Id - F} + \sqrt{F} \cdot X \cdot \sqrt{F}= \luders^\dagger_F(X).
\end{align*}
Using the fact that for any $x \in [-1, 1]$,
\begin{equation*}
    e^{i \arcsin(\sqrt{x})}
    = \cos(\arcsin(\sqrt{x})) + i \sin(\arcsin(\sqrt{x}))
    = \sqrt{1 - x} + i \sqrt{x},
\end{equation*}
the unitary $U$ can also be expressed as $U = e^{i\alpha(\lift{A})}$, for $\alpha = \arcsin(\sqrt{f_{\lambda,\theta}})$.

\begin{lemma}\label{lem:linalg}
    Suppose there are constants $0 < f_- \le f_+ \le 1$ for which a Hermitian operator $F$ satisfies
    \begin{equation}
        f_- \Id \preceq F \preceq f_+\Id.
    \end{equation}
    Let $U$ be a unitary commuting with $F$, and suppose further that $(U-\Id)^{\dagger}(U-\Id) \preceq aF$ for some $a\ge 0$. Then for any Hermitian observable $0 \preceq X \preceq \Id$,
    \begin{equation}
        \frac{1}{2}(U^\dagger X U + UXU^\dagger) - X \preceq a\parens*{1 + \sqrt{\frac{f_+}{f_-}}} F\,.
    \end{equation}
\end{lemma}

\begin{proof}
    For simplicity, we will write
    \begin{equation*}
        D = U - \Id,
        \qquad
        G = D^{\dagger}D = DD^{\dagger}
        = 2\Id - U - U^{\dagger}.
    \end{equation*}
    Note that $D + D^\dagger = -G$, and
    \begin{equation}
        \frac{1}{2}(U^\dagger X U + UXU^\dagger) - X = -\frac{1}{2}(GX + XG) + \frac{1}{2}(D^\dagger XD + DXD^\dagger)\,.
    \end{equation}
    We bound the second term using $D^\dagger XD \preceq D^\dagger D = G \preceq aF$, and $DXD^\dagger \preceq DD^\dagger = G \preceq aF$. For the first term, we use Cauchy--Schwarz:
    \begin{align}
        |\bra{v}GX\ket{v}|
        &\le \sqrt{\bra{v}G\ket{v}\cdot \bra{v}XGX\ket{v}} \\
        &\le a\sqrt{\bra{v}F\ket{v}\cdot \bra{v}XFX\ket{v}}
        \le a\sqrt{\frac{f_+}{f_-}} \cdot \bra{v} F\ket{v}
        \,.
    \end{align}
    The final inequality follows as $XFX \preceq f_+ X^2 \preceq f_+\Id \preceq \frac{f_+}{f_-}F$.
\end{proof}

\begin{lemma}\label{lem:heating}
    Let $F = f(\lift{A})$ for some non-decreasing, differentiable $f: [0,1]\to(0,1]$. Then
    \begin{equation}
        \luders^\dagger_F(\mathsf{N}) - \mathsf{N} \preceq \frac{R_f(1 + \sqrt{R_f})L^2_f}{n}\cdot F\,,
    \end{equation}
    where
    \begin{equation}
        L_f = \sup_{x\in[0,1]} \frac{\alpha'(x)}{\sqrt{f(x)}} \qquad R_f = \sup_{\substack{x,y\in[0,1] \\ 0 \le y - x \le 1/n}} \frac{f(y)}{f(x)}
    \end{equation}
    for $\alpha = \arcsin(\sqrt{f})$.
\end{lemma}

\begin{proof}
Recall that $\mathsf{N} = \sum^n_{j=1} Q_j$.
    For a fixed mode $j$, we will prove the bound
    \begin{equation}\label{eq:goal-bound}
        \luders^\dagger_F(Q_j) - Q_j \preceq \frac{ R_f (1 + \sqrt{R_f})L^2_f}{n^2} \cdot F\,.
    \end{equation}
    This will imply the lemma by linearity. 
    For simplicity, we will show the $j = n$ case, which is without loss of generality.
    To begin, let us suppose that the observable $A$ has eigendecomposition
    \begin{equation*}
        A = \sum_{i=1}^d a_i \cdot \ketbra{v_i}{v_i},
    \end{equation*}
    so that it has eigenvalues $0 \leq a_i \leq 1$ with corresponding eigenvectors $\ket{v_i}$.
    Then we can write the lifted observable as
    \begin{align*}
        \lift{A}
        &= \sum_{i_1, \ldots, i_n = 1}^d \Big(\frac{a_{i_1} + \cdots + a_{i_n}}{n}\Big)\cdot \ketbra{v_{i_1}, \ldots, v_{i_n}}{v_{i_1}, \ldots, v_{i_n}}\\
        &= \sum_{i_{< n} \in [d]^{n-1}}\sum_{i_n=1}^d \Big(\frac{a_{i_{<n}} + a_{i_n}}{n}\Big)\cdot \ketbra{v_{i_{<n}}, v_{i_n}}{v_{i_{<n}}, v_{i_n}},
    \end{align*}
    where we write $i_{< n} = (i_1, \ldots, i_{n-1})$, $v_{i_{<n}} = (v_{i_1}, \ldots, v_{i_{n-1}})$, and $a_{i_{< n}} = a_{i_1} + \cdots +  a_{i_n}$.
    As a result, $\lift{A}$ can be viewed as a block-diagonal matrix, with blocks specified by strings $i_{<n}$:
    \begin{equation*}
        \lift{A}
         = \sum_{i_{< n} \in [d]^{n-1}}\ketbra{v_{i_{<n}}}{v_{i_{<n}}} \otimes \Big(\sum_{i_n=1}^d \Big(\frac{a_{i_{<n}} + a_{i_n}}{n}\Big)\cdot \ketbra{v_{i_n}}{v_{i_n}}\Big)
         \eqqcolon\sum_{i_{< n} \in [d]^{n-1}}\ketbra{v_{i_{<n}}}{v_{i_{<n}}} \otimes \lift{A}_{i_{<n}}.
    \end{equation*}
    As a result, we have that
    \begin{align*}
        F
        &= f(\lift{A})
        = \sum_{i_{< n} \in [d]^{n-1}}\ketbra{v_{i_{<n}}}{v_{i_{<n}}} \otimes f(\lift{A}_{i_{<n}}),\\
        U
        &= e^{i\alpha(\lift{A})}
        = \sum_{i_{< n} \in [d]^{n-1}}\ketbra{v_{i_{<n}}}{v_{i_{<n}}} \otimes e^{i \alpha(\lift{A}_{i_{<n}})}
        \eqqcolon \sum_{i_{< n} \in [d]^{n-1}}\ketbra{v_{i_{<n}}}{v_{i_{<n}}} \otimes U_{i_{<n}}
    \end{align*}
    are also block diagonal.
    In addition, $Q_n = I^{\otimes n-1} \otimes Q$ is block diagonal by definition.
    Note that by Equation~\eqref{eq:logm-channel-time-evolution}, we have
    \begin{equation}\label{eq:blog-diag}
        \luders^\dagger_F(Q_j) = \frac{1}{2}(UQ_nU^\dagger + U^\dagger Q_n U)
        = \sum_{i_{<n} \in [d]^{n-1}} \ketbra{v_{i_{<n}}}{v_{i_{<n}}} \otimes \frac{1}{2}(U_{i_{<n}} Q U_{i_{<n}} ^\dagger + U_{i_{<n}} ^\dagger Q U_{i_{<n}} ).
    \end{equation}

    Fix any block $i_{<n}$, and consider the matrix within that block, given by
    \begin{equation*}
        \lift{A}_{i_{< n}} = \sum_{i_n=1}^d \Big(\frac{a_{i_{<n}} + a_{i_n}}{n}\Big)\cdot \ketbra{v_{i_n}}{v_{i_n}}
        \eqqcolon \sum_{i_n=1}^d b_{i_n} \cdot \ketbra{v_{i_n}}{v_{i_n}}.
    \end{equation*}
    Let $i_-$ and $i_+ \in [d]$ be such that $b_{i_-}$ is the smallest eigenvalue of $\lift{A}_{i_{<n}}$ and $b_{i_+}$ is its largest eigenvalue.
    Since each $a_{i_n}$ is in $[0, 1]$, we must have $b_{i_+} - b_{i_-} \leq 1/n$.
    Thus, if we consider the interval $\mathcal{I} = [b_{i_-}, b_{i_+}]$, then $\mathcal{I}$ contains every eigenvalue of $\lift{A}_{i_{< n}}$, and it has width at most $|\mathcal{I}| \leq 1/n$.
    Next, define $f_- = f(b_{i_-})$ and $f_+ = f(b_{i_+})$.
    Because $f$ is a monotonically increasing function, $f_-$ and $f_+$ are the smallest and largest eigenvalues of $f(\lift{A}_{i_{< n}})$, respectively.
    (Indeed, $f_+$ must be the largest value that $f$ attains on the interval $\mathcal{I}$.)
    Thus, if we set $R_0 = f_+ / f_-$, then it holds that $R_0 \le R_f$.
    Over the interval $\mathcal{I}$, let
    \begin{equation}
        \alpha_1 \triangleq \sup_{x\in\mathcal{I}} \alpha(x),
        \qquad \alpha_0 \triangleq \inf_{y\in\mathcal{I}} \alpha(y),
        \qquad \Delta \triangleq \alpha_1 - \alpha_0.
    \end{equation} 
    Note that
    \begin{equation}\label{eq:delta-inequality}
        \Delta
        \le \frac{1}{n}\cdot \sup_{x\in\mathcal{I}}\alpha'(x)
        \le \frac{1}{n}\sqrt{f_+}\cdot \sup_{x\in\mathcal{I}}\frac{\alpha'(x)}{\sqrt{f(x)}}
        \le \frac{L_f}{n}\sqrt{f_+}\,,
    \end{equation}
    by the definition of $L_f$.

    Consider the unitary $ U_{i_{<n}}'
        \triangleq U_{i_{<n}} \cdot e^{-i \alpha_0}$. 
    As multiplying by a phase does not affect conjugation, we can rewrite the expression for $\luders^\dagger_F(Q_j)$ within this block
    from \eqref{eq:blog-diag}
    as
    \begin{equation*}
        \frac{1}{2}(U_{i_{<n}} Q U_{i_{<n}} ^\dagger + U_{i_{<n}} ^\dagger Q U_{i_{<n}} )
        =\frac{1}{2}(U_{i_{<n}}' Q U_{i_{<n}}'^\dagger + U_{i_{<n}}'^\dagger Q U_{i_{<n}}').
    \end{equation*}
    Let us now consider the action of $U'_{i_{n}}$ within this block. 
    Within this block, on an eigenvector of $\lift{A}$ with eigenvalue~$x$, it acts by applying the eigenvalue $e^{i(\alpha(x) - \alpha_0)}$. Thus $(U'_{i_{<n}} - \Id)^{\dagger} (U'_{i_{<n}} - \Id)$ acts on this eigenvector by applying the eigenvalue
    \begin{align*}
        |e^{i(\alpha(x) - \alpha_0)} - 1|^2 &\leq |\alpha(x) - \alpha_0|^2 \tag*{\text{(by $|e^{it} - 1| \le |t|$)}}\\
        &\leq \Delta_0^2\\
        &\leq \frac{L_f^2 f_+}{n^2} \tag*{\text{(by \eqref{eq:delta-inequality})}}\\
        &\leq \frac{L_f^2 f_+}{n^2 f_-} \cdot f(x) \tag*{\text{(by definition of $f_-$)}}\\
        &= \frac{L_f^2 R_0}{n^2} \cdot f(x) \tag*{\text{(by definition of $R_0$)}}\\
        &\leq \frac{L_f^2 R_f}{n^2} \cdot f(x). \tag*{\text{(by $R_0 \leq R_f$)}}
    \end{align*}
    Thus, within this block, we have that
    \begin{equation}
        (U'_{i_{<n}} - \Id)^{\dagger} (U'_{i_{<n}} - \Id)
        \preceq \frac{L^2_f R_f}{n^2} F\,.
    \end{equation}
    We can thus apply Lemma~\ref{lem:linalg} with $X = Q$ and $a = \frac{L^2_f R_f}{n^2}$ to get
    \begin{equation}
        \frac{1}{2}(U_{i_{<n}} Q U_{i_{<n}} ^\dagger + U_{i_{<n}} ^\dagger Q U_{i_{<n}} ) - Q
        \preceq \frac{L^2_f R_f (1 + \sqrt{R_0})}{n^2} \cdot F
        \preceq \frac{L^2_f R_f (1 + \sqrt{R_f})}{n^2} \cdot F.
    \end{equation}
    As this holds over all blocks, we  have that
    \begin{equation}
        \luders^\dagger_F(Q_n) - Q_n \preceq \frac{L^2_f R_f (1 + \sqrt{R_f})}{n^2} \cdot F\,.
    \end{equation}
    This relation holds over all blocks.
    This therefore holds for all $Q_j$ as well, and summing over all $n$ modes $Q_j$ yields the desired claim.
\end{proof}

\begin{proof}[Proof of Corollary~\ref{cor:damage}]
    Taking $f$ in the above lemmas to be $f_{b, \lambda}(x) = 1/(1 + e^{\lambda(b - x)})$, we will compute below that $L_f \le \lambda/2$ and $R_f \le e^{\lambda/n}$. This is the only part of the proof where we use the specific form of $f$; elsewhere, we only needed that it was nonnegative, nondecreasing, and had range $[0,1]$.
    
    First, for $f(x) = 1/(1 + e^{-\lambda(x - b)})$, note that $f'(x) = \lambda f(x)(1 - f(x))$, so for $\alpha = \arcsin(\sqrt{f})$,
    \begin{equation}
        \alpha'(x) = \frac{f'(x)}{2\sqrt{f(x)(1 - f(x))}} = \frac{\lambda}{2}\sqrt{f(x)(1 - f(x))}
    \end{equation}
    and thus 
    \begin{equation}
        \frac{\alpha'(x)}{\sqrt{f(x))}} = \frac{\lambda}{2}\sqrt{1 - f(x)}\le \frac{\lambda}{2}\,.
    \end{equation}
    Next, we compute $R_f$.
    First, note that
    \begin{equation*}
        (\log f)'(x)
        = \frac{f'(x)}{f(x)} = \lambda (1 - f(x)) \leq \lambda.
    \end{equation*}
    As a result, for all $x \leq y$, we can bound
    \begin{equation*}
        \log(f(y)/f(x))
        = \log f(y) - \log f(x)
        = \int_{x}^y (\log f)'(z) \cdot dz
        \leq \int_{x}^y \lambda \cdot  dz
        = \lambda (y-x),
    \end{equation*}
    and so $f(y)/f(x) \leq e^{\lambda(y-x)}$.
    Thus, over $0 \le y - x\le 1/n$ we conclude that $R_f \le e^{\lambda/n}$.
    Substituting these into Lemma~\ref{lem:heating}, we conclude that
    \begin{equation}
        \luders^\dagger_F(\mathsf{N}) - \mathsf{N} \preceq \frac{e^{\lambda/n}(1 + e^{\lambda/2n})\lambda^2/4}{n}\cdot F\,.
    \end{equation}
    Because $0 < \lambda \le n$ by assumption, $\luders^\dagger_F(\mathsf{N}) - \mathsf{N} \preceq \frac{2\lambda^2}{n} F$, and the claim follows by taking trace with $\tau$ on both sides and dividing by $n$.
\end{proof}

\section{Dimension-free rate}

In this section we prove the upper bound in Theorem~\ref{thm:main_sqrtm}, restated here for convenience:

\sqrtm*

\subsection{Description of protocol}

Unlike the protocol in Section~\ref{sec:logm}, the protocol here does not use a complicated matrix multiplicative weights outer wrapper. Instead, we simply process the adversary's observables in sequence, and for each $A^{(t)}$, we consider the following soft measurement.

\begin{definition}[Measurement with compact noise]
    \label{def:measurement-compact-noise}
    Given a width parameter $\omega > 0$, define the kernel
    \begin{equation}
        \phi_{\omega}(x) \triangleq \begin{cases}
            \omega^{-1/2}\cos(\pi x / 2\omega) & \text{if} \ |x| \le \omega \\
            0 & \text{otherwise}
        \end{cases}
    \end{equation}
    and let $M^\omega_{\lift{A}}$ denote the positive operator-valued measure with density 
    % \angelos{should the RHS have a $\mathrm{d}y?$}\sitan{yes!}
    \begin{equation}
        \mathrm{d}M^\omega_{\lift{A}}(y) \triangleq \phi_{\omega}(y\Id - \lift{A})^2\,\mathrm{d}y
    \end{equation}
    and associated Kraus operators $K^\omega_{\lift{A}}(y) \triangleq \phi_{\omega}(y\Id - \lift{A})$.
    We will ultimately take 
    \begin{equation}
        \omega = \epsilon/4\,.
    \end{equation}
    Given a pure state $\ket{\phi}$, measuring with $M^\omega_{\lift{A}}$ results in the convolution of $\phi_\omega$ with the distribution given by measuring in the eigenbasis of $\lift{A}$. The sub-normalized post-measurement state is $K^\omega_{\lift{A}}\ket{\phi}$.
\end{definition}
Note that this is a valid POVM because for any eigenvector $\ket{v}$ of $\lift{A}$ with eigenvalue $x$, we have
\begin{equation}
    \bra{v} \Big(\int \mathrm{d}M^\omega_{\lift{A}}(y)\Big) \ket{v}
    = \int_y \phi_{\omega}(y - x)^2 \,\mathrm{d}y
    = \int_{x-\omega}^{x+\omega} \omega^{-1} \cos(\pi (y-x)/2\omega)^2 \,\mathrm{d}y
    = \int_{-1}^{1} \cos(\pi u/2)^2 \,\mathrm{d}u = 1,
    \label{eq:compact-kernel-integral}
\end{equation}
where we have used the substitution $u = (y - x)/\omega$.

\noindent The pseudocode for the protocol is given in Algorithm~\ref{alg:shadow-sqrtm}.

\begin{algorithm}[H]
\DontPrintSemicolon
\caption{\textsc{OnlineShadowTomography2}}
\label{alg:shadow-sqrtm}
\SetKw{KwAnd}{and}
\SetKw{KwOr}{or}
\SetKwInOut{Params}{Parameters}
\SetKwInOut{StateVar}{State}
\SetKwComment{Cmt}{$\triangleright$~}{}
\Params{target accuracy $\epsilon$;\ \ number of rounds $m$.}
\For{$t = 1, 2, \dots, m$}{
  $\omega \gets \epsilon/4$.\;
  Receive observable $A^{(t)}$, chosen by the adversary as a function of the transcript $\calT_{t-1}$.\;
  
  Apply the measurement $M^\omega_{\lift{A}^{(t)}}$ to the current $n$-copy register, obtaining outcome $y$.\;
  Output $\widehat{\mu}^{(t)} \gets y$\;
}
\end{algorithm}

\noindent Our notion of badness is as follows:

\begin{definition}
    We say that a round $t$ is \emph{bad} if the corresponding measurement with compact noise returns an estimate of $\Tr(A^{(t)}\rho)$ which is not within $\epsilon$ of the true answer. The learner does not know if a given round is bad, and this terminology is only used for the analysis.
\end{definition}

\subsection{Bad outcomes are unlikely}

We use a similar charging argument as in the previous section. Consider the tree of possible complete internal histories of the protocol, with each node labeled by the sub-normalized conditional state at that point in the protocol. Note that every internal node of this tree now has a continuum of children, rather than just two. Each leaf node plays the same role as before, either
\begin{itemize}
    \item The first bad outcome that occurs in that branch of the tree, or
    \item the final outcome in the protocol, where all outcomes along that branch were non-bad,
\end{itemize}
except that our notion of bad is different from Section~\ref{sec:logm}.
Additionally, the depth of this tree is at most $m$ by design.

Let $\mathcal{D}$ denote the continuum of edges of the tree corresponding to the first bad outcome occurring along that path. For each such edge $e\in\mathcal{D}$, let $F_e$ denote the measurement operator with compact noise corresponding to the bad outcome, and let $\sigma_e$ (resp.\ $\tau_e$) denote the sub-normalized post-measurement state right before (resp.\ after) the bad outcome.

As before, we wish to bound the sum over all edges $e\in\mathcal{D}$ of $\Tr(\tau_e)$, as this is the probability over the course of the entire protocol that a bad outcome happens. However, because the children of each node in the tree are a continuum indexed by $\mathbb{R}$, we have to be careful about the notion of summation. We write this as
\begin{equation}
    \int_{\mathcal{D}} \Tr(\tau_e)\,\mathrm{d}e\,,
\end{equation}
where $\mathrm{d}e$ denotes the sum over transcript lengths $\ell$ of the Lebesgue measure over length-$\ell$ transcripts, which are naturally indexed by $\mathbb{R}^\ell$, restricted to transcripts ending with edge $e$. We will use similar shorthand later when referring to integration over other sets of nodes in the tree.
%\sitan{is this okay?}

We implement the same charging argument as before. The key technical lemma is the following analogue of Corollary~\ref{cor:damage}:

\begin{lemma}\label{lem:sqrtm_perstep}
    Given $y\in\mathbb{R}$, let $F_y\triangleq \phi_\omega(y\Id - \lift{A})$. Let $\sigma$ be an arbitrary mixed state, and let
    \begin{equation}
        \luders(\sigma) \triangleq \int F_y \sigma F_y\,\mathrm{d}y
    \end{equation}
    denote the expected post-measurement state. Then
    % \ryan{gpt claims the rhs should have a factor of $\Tr(\sigma)$.}
    \begin{equation}
        \energy[\luders(\sigma)] - \energy[\sigma] \le \frac{\pi^2}{2n^2\omega^2}\Tr(\sigma)\,.
    \end{equation}
\end{lemma}

\noindent We defer the proof of this to Section~\ref{sec:sqrtm_perstep_proof} and proceed to formally present the argument outlined above. In analogy with Lemma~\ref{lem:selfbound1}, we need the following simple consequence of the formalism from Section~\ref{sec:pickl} which bounds the probability of reaching a particular bad leaf in terms of the energy of the corresponding sub-normalized conditional state.

\begin{lemma}\label{lem:selfbound2}
    Suppose that at some point in the protocol {\sc OnlineShadowTomography2}, the adversary queries observable $A$, and after performing the measurement with compact noise $M^\omega_{\lift{A}}$ with $\omega = \epsilon/4$ we obtain outcome $y$ for which $|y - \Tr(A\rho)| > \epsilon$. If $\tau$ denotes the sub-normalized state in the tree right after this measurement outcome, and if $n\epsilon^2 > 32$, then
    \begin{equation}
        \Tr(\tau) \le \frac{32}{\epsilon^2}\energy[\tau]\,.
    \end{equation}
\end{lemma}

\begin{proof}
    Because the outcome $y$ is inaccurate, $\tau$ is supported on the quantum event
    \begin{equation}
        |\lift{A} - \Tr(A\rho)\Id| \ge |y - \Tr(A\rho)| - \epsilon/4 > 3\epsilon/4\,.
    \end{equation}
    and is thus preserved under the projector $\mathbf{1}_{|x - \Tr(A\rho)| \ge 3\epsilon/4}$. So the claim follows by Lemma~\ref{lem:secondmoment}.
\end{proof}

\begin{proof}[Proof of Theorem~\ref{thm:main_sqrtm}]
    As in the proof of Lemma~\ref{lem:badbound}, we pass to an integral over \emph{all} leaves of the tree. Telescoping and noting again that $\energy[\tau_{\rm root}] = 0$ for the root of the tree, we have
    \begin{equation}
        \int \energy[\tau]\,\mathrm{d}\tau = \int \energy[\tau]\,\mathrm{d}\tau- \energy[\tau_{\rm root}] = \int_{\rm interior} (\energy[\luders_\sigma(\sigma)] - \energy[\sigma])\,\mathrm{d}\sigma \le \frac{\pi^2}{2n^2\omega^2}\int_{\rm interior} \Tr(\sigma)\,\mathrm{d}\sigma\,,
    \end{equation}
    where the $\luders_\sigma(\sigma)$ denotes the post-measurement state after performing the POVM $M^\omega_{\lift{A}}$ on interior node $\sigma$.
    Note that because the tree is of depth $m$ and the integral of $\Tr(\sigma)$ over sub-normalized conditional states at any particular layer is at most~$1$, we conclude by our choice of $\omega = \epsilon/4$ that
    \begin{equation}
        \int_{\mathcal{D}}\energy[\tau_e]\,\mathrm{d}e \le \int \energy[\tau]\,\mathrm{d}\tau \le \frac{8\pi^2 m}{n^2\epsilon^2}\,, \label{eq:charge2}
    \end{equation}
    By Lemma~\ref{lem:selfbound2}, the left-hand side is at least $\frac{\epsilon^2}{32}\int_{\mathcal{D}} \Tr(\tau_e)\,\mathrm{d}e$.
    Rearranging, we conclude that the probability that there is any bad outcome, i.e., $\int_{\mathcal{D}} \Tr(\tau_e)\,\mathrm{d}e$, is at most $\mathcal{O}(\frac{m}{n^2\epsilon^4})$, from which the theorem follows.
\end{proof}

\subsection{Per-step damage control: Proof of Lemma~\ref{lem:sqrtm_perstep}}
\label{sec:sqrtm_perstep_proof}

%\angelos{Let me know if someone had a different proof in mind. I tried to mimic Section 5.3's proof.}
We first introduce some useful notation and facts about the Fourier transform of the compact noise kernel $\phi_{\omega}$. In particular, we use $\widehat{\phi}_\omega$ for the unitary Fourier transform of $\phi_\omega$
\begin{equation*}
    \widehat{\phi}_\omega(t) = \frac{1}{\sqrt{2\pi}} \int e^{-itz} \cdot \phi_{\omega}(z)\,\mathrm{d}z.
\end{equation*}
\begin{claim}
    \label{cl:compact-kernel-fourier-facts}
   Let $\phi_{\omega}$ be the kernel from Definition~\ref{def:measurement-compact-noise} and $q_\omega(t) = |\widehat{\phi}_\omega(t)|^2$. Then
   \begin{equation*}
        \int q_\omega(t)\,\mathrm{d}t = 1,
        \qquad
        \int t^2 q_\omega(t)\,\mathrm{d}t = \frac{\pi^2}{4\omega^2}.
   \end{equation*}
\end{claim}

\begin{proof}
    Plancherel along with Eq.~\eqref{eq:compact-kernel-integral} implies that
    \begin{equation*}
        \int q_\omega(t)\,\mathrm{d}t
        = \int |\widehat{\phi_{\omega}}(t)|^2\,\mathrm{d}t
        = \int |\phi_{\omega}(t)|^2\,\mathrm{d}t
        = \int_{-\omega}^{\omega} \omega^{-1} \cos(\pi t/2\omega)^2\,\mathrm{d}t
        = 1
    \end{equation*}
    Moreover,
    \begin{align}
        \int t^2 q_\omega(t)\,\mathrm{d}t
        &= \int t^2 |\widehat{\phi_\omega}(t)|^2\,\mathrm{d}t \\
        &= \int |it \cdot \widehat{\phi_\omega}(t)|^2\,\mathrm{d}t \\
        &= \int |\widehat{\phi_\omega'}(t)|^2\,\mathrm{d}t 
        = \|\phi_{\omega}'\|_2^2.
    \end{align}
    The second equality follows because differentiation becomes multiplication by $it$ under the Fourier transform. The derivative of $\phi_{\omega}$ is now \angelos{it says weak derivative, ChatGPT says that it is valid}
    \begin{equation*}
        \phi_{\omega}'(x) = -\frac{\pi}{2\omega^{3/2}}\sin(\pi x/2\omega) \cdot \mathbf{1}_{|x| < \omega}.
    \end{equation*}
    Then
    \begin{equation*}
        \|\phi_{\omega}'\|_2^2
        = \frac{\pi^2}{4\omega^{3}} \int_{-\omega}^{\omega} \sin(\pi x/2\omega)^2 \, \mathrm{d}x
        = \frac{\pi^2}{4\omega^{2}} \int_{-1}^{1} \sin(\pi u/2)^2 \, \mathrm{d}u
        = \frac{\pi^2}{4\omega^2},
    \end{equation*}
    where we have used the substitution $u = x/\omega$.
\end{proof}

\begin{proof}[Proof of Lemma~\ref{lem:sqrtm_perstep}]
    We will prove the statement in the Heisenberg picture by showing that
    \begin{equation*}
        \luders^{\dagger}(\mathsf{N}) - \mathsf{N}
        \preceq \frac{\pi^2}{2n\omega^2} \cdot \Id.
    \end{equation*}
    This implies the desired inequality by
    \begin{equation*}
        \energy[\luders(\sigma)] - \energy[\sigma]
        = \frac{1}{n}\Big( \Tr(\mathsf{N} \cdot \luders(\sigma)) - \Tr(\mathsf{N} \cdot \sigma)\Big)
        = \frac{1}{n} \cdot \Tr((\luders^{\dagger}(\mathsf{N}) - \mathsf{N}) \cdot \sigma)
        \leq \frac{\pi^2}{2n^2\omega^2} \Tr(\sigma).
    \end{equation*}
    We show in Claim~\ref{cl:sqrtm-channel-as-time-evolution} that for $U_t = e^{it\lift{A}}$, the adjoint of the channel that maps to the post-measurement state can be written as
    \begin{equation*}
        \luders^{\dagger}(X)
        = \int q_\omega(t) \cdot \frac{1}{2} \big(U_t^{\dagger} X U_t + U_t X U_t^{\dagger}\big) \,\mathrm{d}t.
    \end{equation*}
    Recall that $\mathsf{N} = \sum_{j=1}^n Q_j$. By linearity, it suffices to bound $\luders^{\dagger}(Q_j)$ for a single mode $j$.
    It holds that
    \begin{equation*}
        U_t^{\dagger} Q_j U_t
        = e^{-it\lift{A}} \cdot Q_j \cdot e^{it\lift{A}}
        = \Big(\prod_{j=1}^n e^{-itA_j/n}\Big) Q_j \Big(\prod_{j=1}^n e^{itA_j/n}\Big)
        = e^{-itA_j/n} \cdot Q_j \cdot e^{itA_j/n}.
    \end{equation*}
    Applying Lemma~\ref{lem:linalg} with $F$ as the identity implies that a unitary $U$ satisfies
    \begin{equation*}
        \frac{1}{2} \big(U^{\dagger} X U + U X U^{\dagger}\big) - X \preceq 2a\Id,
    \end{equation*}
    whenever $(U - \Id)^{\dagger}(U - \Id) \preceq a\Id$. In our case, since the spectrum of $A_j$ lies in $[0, 1]$ and using $|e^{ix} - 1| \leq |x|$, it holds that
    \begin{equation*}
        (e^{itA_j/n} - \Id)^{\dagger}(e^{itA_j/n} - \Id) = |e^{itA_j/n} - \Id|^2 \preceq \frac{t^2}{n^2} \cdot \Id,
    \end{equation*}
    which means
    \begin{equation}
        \frac{1}{2} \big(U_t^{\dagger} Q_j U_t + U_t Q_j U_t^{\dagger}\big) - Q_j
        = \frac{1}{2} \big(e^{-itA_j/n} Q_j e^{itA_j/n} + e^{itA_j/n} Q_j e^{-itA_j/n}\big) - Q_j
        \preceq \frac{2t^2}{n^2} \cdot \Id.
        \label{eq:channel-psd-upperbound-sqrtm}
    \end{equation}
    We now write using Equation~\eqref{eq:channel-psd-upperbound-sqrtm} and Claim~\ref{cl:compact-kernel-fourier-facts}
    \begin{align*}
        \luders^{\dagger}(Q_j) - Q_j
        &= \int q_\omega(t) \cdot \bigg(\frac{1}{2} \big(U_t^{\dagger} Q_j U_t + U_t Q_j U_t^{\dagger}\big) - Q_j\bigg)\,\mathrm{d}t \\
        &\preceq \Big(\frac{2}{n^2}\int t^2 q_\omega(t) \,\mathrm{d}t\Big) \cdot \Id
        = \frac{\pi^2}{2n^2\omega^2} \cdot \Id.
    \end{align*}
    The desired result follows by linearity:
    \begin{equation*}
        \luders^{\dagger}(\mathsf{N}) - \mathsf{N}
        = \sum_{j=1}^n \big(\luders^{\dagger}(Q_j) - Q_j\big)
        \preceq n \cdot \frac{\pi^2}{2n^2\omega^2} \cdot \Id = \frac{\pi^2}{2n\omega^2} \cdot \Id. \qedhere
    \end{equation*}
\end{proof}
We prove the final ingredient from the proof below.
\begin{claim}
    \label{cl:sqrtm-channel-as-time-evolution}
    Given $y\in\mathbb{R}$ and PSD observable $A$, let $F_y = \phi_{\omega}(y\Id - \lift{A})$ and $U_t = e^{it\lift{A}}$. Then
    \begin{equation}
        \int_y F_y X F_y\,\mathrm{d}y
        = \int q_\omega(t) \cdot \frac{1}{2} \big(U_t^{\dagger} X U_t + U_t X U_t^{\dagger}\big)\,\mathrm{d}t.
    \end{equation}
\end{claim}

\begin{proof}
    We write the spectral decomposition of $\lift{A} = \sum_{\lambda} \lambda\cdot \Pi_{\lambda}$.
    Then
    \begin{equation*}
        F_y = \phi_{\omega}(y\Id - \lift{A})
        = \sum_{\lambda} \phi_{\omega}(y - \lambda) \cdot \Pi_{\lambda}.
    \end{equation*}
    Substituting
    \begin{align*}
        \int_y F_y X F_y\,\mathrm{d}y 
        &= \int_y \big(\sum_{\lambda} \phi_{\omega}(y - \lambda) \cdot \Pi_{\lambda}\big) X \big(\sum_{\mu} \phi_{\omega}(y - \mu) \cdot \Pi_{\mu}\big)\,\mathrm{d}y \\
        &= \sum_{\lambda, \mu} \Big(\int_y \phi_{\omega}(y - \lambda) \cdot \phi_{\omega}(y - \mu)\,\mathrm{d}y \Big) \cdot \Pi_{\lambda} X \Pi_{\mu} \\
        &= \sum_{\lambda, \mu} \Big(\int_t |\widehat{\phi}_\omega(t)|^2 \cdot e^{it(\lambda-\mu)}\,\mathrm{d}t \Big) \cdot \Pi_{\lambda} X \Pi_{\mu} \\
        &= \int_t q_\omega(t) \Big(\sum_{\lambda} e^{it\lambda} \cdot \Pi_{\lambda} \Big) X \Big(\sum_{\mu} e^{-it\mu} \cdot \Pi_{\mu} \Big) \,\mathrm{d}t \\
        &= \int_t q_\omega(t) \cdot U_t X U_t^{\dagger} \,\mathrm{d}t.
    \end{align*}
    The third inequality follows by defining $f_{\lambda}(y) = \phi_{\omega}(y - \lambda)$ and observing that
    \begin{equation*}
        \widehat{f_{\lambda}}(t) = e^{-it\lambda} \cdot \widehat{\phi_{\omega}}(t).
    \end{equation*}
    Since $\phi_{\omega}$ is a real function we use Plancherel to deduce
    \begin{align}
        \int_y \phi_{\omega}(y - \lambda) \cdot \phi_{\omega}(y - \mu)\,\mathrm{d}y
        &= \int_y \overline{f_{\lambda}(y)} \cdot f_{\mu}(y)\,\mathrm{d}y \\
        &= \int_t \overline{\widehat{f_{\lambda}}(t)} \cdot \widehat{f_{\mu}}(t)\,\mathrm{d}t \\
        &= \int_t e^{it(\lambda - \mu)} \cdot \overline{\widehat{\phi_{\omega}}(t)} \cdot \widehat{\phi_{\omega}}(t)\,\mathrm{d}t \\
        &= \int_t e^{it(\lambda - \mu)} \cdot |\widehat{\phi_{\omega}}(t)|^2\,\mathrm{d}t.
    \end{align}
    Finally, since $q_\omega$ is an even function, we can average over $t$ and $-t$ to get
    \begin{equation*}
        \luders^{\dagger}(X)
        = \int q_\omega(t) \cdot \frac{1}{2} \big(U_t^{\dagger} X U_t + U_t X U_t^{\dagger}\big) \,\mathrm{d}t. \qedhere
    \end{equation*}
\end{proof}

\section*{Acknowledgments}

The authors are grateful to the numerous individuals with whom they have shared inspiring discussions about Shadow Tomography, including Costin B{\u a}descu, Ainesh Bakshi, John Bostanci, Weiyuan Gong, Jerry Li, Allen Liu, Jack Spilecki, Ewin Tang, Qi Ye, and Zhihan Zhang.

\paragraph{Statement on AI use.} The main ideas underlying this work, including the central proof strategies, were generated by ChatGPT 5.6-Sol Pro. The human authors carefully studied, refined, and verified these ideas, and substantially rewrote and reorganized their presentation, adding the motivation and exposition necessary to communicate the arguments clearly. The final paper reflects the authors’ own understanding of the results, and the authors take full responsibility for every claim, proof, and citation contained in it.

\bibliographystyle{alpha}
\bibliography{refs}

\appendix

\section{Optimal Threshold Search by information theory methods}
In this appendix we show possibly the shortest route to obtaining $\log(m)\cdot \poly(\log d, 1/\eps)$ dependence in \emph{Offline} Shadow Tomography.
The method is information-theoretical, and (seemingly) unrelated to the methods used in the rest of this paper.

We consider the following quantum statistical decision theory problem.  It is an ``argmax'' version of Shadow Tomography, more or less equivalent to the ``Threshold Search'' problem from \cite{BO24}:
\begin{definition}
    The (worst-case) \emph{Best Observable Selection} problem is as follows.  Given are (classical descriptions of) observables $0 \preceq A_1, \dots, A_m \preceq \Id$ on $\CC^d$, as well as $n$ copies of an unknown mixed state $\rho \in \CC^{d \times d}$.  The task is to measure the copies and output $\bj \in [m]$ with small \emph{regret},
    \begin{equation}
        r_\rho(\bj) = \mu_* - \mu_{\bj},
    \end{equation}
    where $\mu_i \coloneqq \Tr(\rho A_i)$ and $\mu_* \coloneqq \max_i \{\mu_i\}$.
\end{definition}

\noindent We prove:

\begin{theorem} \label{thm:BOS}
    There is  a Best Observable Selection algorithm with expected regret at most $\sqrt{\frac{\ln m}{2n}}$. Thus the expected regret is at most $\epsilon$ provided $n \geq \frac{\ln m}{2\eps^2}$.
\end{theorem}

\noindent We may view the observables $A_1, \dots, A_m$ as defining a two-player game between an Adversary and a Learner. 
Simultaneously, the Adversary chooses a state~$\brho$ (possibly at random), and the Learner chooses a decision rule POVM $M = (M_1, \dots, M_m)$ on $(\CC^{d})^{\otimes n}$.  
The ``payoff'' is the expected regret $\EE_{\bj \sim M(\brho^{\otimes n})}[r(\bj)]$.  
By standard minimax theory for bilinear games on convex compact sets~\cite{Sio58}, the game has an optimal value that is equally achieved when the Adversary must play first (with a randomized strategy) and when the Learner must play first.
The Learner-first version corresponds to the Best Observale Selection scenario; thus for Theorem~\ref{thm:BOS} we may equivalently study the Adversary-first scenario.  
This Adversary-first scenario is a Bayesian one, where $\brho$ is drawn from a prior distribution $\cal{D}$, known to the Learner.  
Thus to prove Theorem~\ref{thm:BOS}, it is necessary and sufficient to prove the following:
\begin{theorem} \label{thm:BBOS}
    Consider a Bayesian version of Best Observable Selection, with input including a known distribution $\calD$ on states.
    Then there is a POVM $M$ (depending only on $\calD$ and $A_1, \dots A_m$) satisfying
    \begin{equation}
        \EE_{\substack{\brho \sim \calD \\ \bj \sim M(\brho^{\otimes n})}}[r_{\brho}(\bj)] \leq \sqrt{\tfrac{\ln m}{2n}}.
    \end{equation}
\end{theorem}
\begin{proof}
    Given $\brho \sim \calD$, let $\bj^* = \bj^*(\brho)$ denote the index (unknown to the Learner) achieving $\bmu_* = \max_j\{\Tr(\brho A_j)\}$.
    For the sake of analysis, consider the classical-quantum state $\ketbra{\bj^*}{\bj^*} \otimes \brho^{\otimes n}$, where we call the first (classical) register $\mathsf{J}$ and the remaining registers $\mathsf{B}_1, \dots, \mathsf{B}_n$.  
    We then introduce the averaged state
    \begin{equation}
        \Omega_{\mathsf{J}\mathsf{B}_1\cdots \mathsf{B}_n} = \EE_{\brho}[\ketbra{\bj^*}{\bj^*} \otimes \brho^{\otimes n}].
    \end{equation}

    We will consider Learner POVMs that happen to depend on only the first $T$ copies of~$\brho$, for $0 \leq T < n$.  
    With such learners in mind, we could define the (conditional mutual) information that seeing copy~$T+1$ would reveal about $\bj^*$:
    \begin{equation}
        \eta(T) \coloneqq I(\mathsf{J} : \mathsf{B}_{T+1} \mid \mathsf{B}_1 \cdots \mathsf{B}_{T})_{\Omega}.
    \end{equation}
    By the chain rule, 
    \begin{equation}
        \sum_{T=0}^{n-1} \eta(T) = I(\mathsf{J} : \mathsf{B}_1 \cdots \mathsf{B}_{n})_{\Omega} = H(\mathsf{J})_\Omega - H(\mathsf{J} \mid \mathsf{B}_1 \cdots \mathsf{B}_{n})_\Omega \leq \ln m,
    \end{equation}
    where the last step used that $\mathsf{J}$ is classical and supported on~$[m]$.
    Thus there is some particular $t$ with $\eta(t) \leq \frac{\ln m}{n}$; this will be the good number of copies for the Learner to use.

    By the Fawzi--Renner Theorem~\cite{FR15}, there exists a channel $\calR : \mathsf{B}_1 \cdots \mathsf{B}_{t}\to \mathsf{J}\mathsf{B}_1 \cdots \mathsf{B}_{t}$ that, when applied to $\Omega_{\mathsf{B}_1 \cdots \mathsf{B}_{t+1}}$ (without touching register $\mathsf{B}_{t+1}$), yields a state $\wt{\Omega}_{\mathsf{J}\mathsf{B}_1\cdots\mathsf{B}_{t+1}}$ satisfying 
    \begin{equation} \label{eqn:dkl}
        \DKLmeas{\Omega_{\mathsf{J}\mathsf{B}_1\cdots\mathsf{B}_{t+1}}}{\wt{\Omega}_{\mathsf{J}\mathsf{B}_1\cdots\mathsf{B}_{t+1}}} \leq \eta(t) \leq \tfrac{\ln m}{n}.
    \end{equation}
    (Here we actually use the sharpened form of the theorem with measured relative entropy, rather than $-2 \ln \mathrm{F}$, due to Brand\~ao et~al.~\cite{BHOS15}.)
    The Learner will apply this channel to its first~$t$ copies of~$\brho$; for analysis purposes, we can think of the $\mathsf{B}_{t+1}$ register as coming along as a ``holdout'' used to analyze regret.
    Then, the learner will measure the register $\mathsf{J}$, obtaining and outputting~$\wt{\bj}$.

    Now for analysis purposes, consider the process of measuring $\mathsf{J}$ to obtain some $\bj \in [m]$, then measuring $\mathsf{B}_{t+1}$ with effect $A_{\bj}$ to obtain some $\ba \in \{0,1\}$.
    On one hand, the ``reward'' $\mu_{\wt{\bj}}$ obtained by the Learner is $\EE[\wt{\ba}]$, for $(\wt{\bj}, \wt{\ba})$ obtained from the measurement process applied to $\wt{\Omega}_{\mathsf{J}\mathsf{B}_1\cdots\mathsf{B}_{t+1}}$.  On the other hand, the best reward $\mu_*$ is $\EE[\ba^*]$, for $(\bj^*, \ba^*)$ obtained from the measurement process applied to $\Omega_{\mathsf{J}\mathsf{B}_1\cdots\mathsf{B}_{t+1}}$.
    Equation~\eqref{eqn:dkl} precisely gives us the classical KL bound $\dKL{(\bj^*,\ba^*)}{(\wt{\bj},\wt{\ba})} \leq \frac{\ln m}{n}$.  Thus 
    \begin{equation}
        \EE[r_{\brho}(\wt{\bj})] = \EE[\ba^*] - \EE[\wt{\ba}] \leq \dtv{\ba^*}{\wt{\ba}} \leq \sqrt{\tfrac12\dKL{\ba^*}{\wt{\ba}}} %\leq \sqrt{\tfrac12 \dKL{(\bj^*,\ba^*)}{(\wt{\bj},\wt{\ba})}} 
        \leq \sqrt{\tfrac{\ln m}{2n}},
    \end{equation}
    where we used Pinsker's inequality and data processing.
\end{proof}
% \begin{remark}
%     The POVM described in the proof never uses more than $n-1$ copies of~$\brho$ (for indexing clarity).
%     Thus the ``$2n$'' appearing in \Cref{thm:BOS,thm:BBOS} can be replaced with ``$2(n+1)$''.
% \end{remark}
\begin{corollary}
    Shadow Tomography can be accomplished with $\log(m) \cdot \wt{O}(\frac{\log d}{\eps^4})$ copies.
\end{corollary}
\begin{proof}
    Applying Markov to the previous result immediately solves the ``$3/4$ vs.~$1/4$'' version of Threshold Search from \cite[Section~4.1]{BO24} with $\mathcal{O}(\log m)$ copies.  Then, Shadow Tomography with $\log(m) \cdot \wt{\mathcal{O}}(\frac{\log d}{\eps^4})$ copies follows by combination with online quantum state learning~\cite{ACH+19}, exactly as in \cite{Aar20,BO24}.
\end{proof}
%We can also give an improved bound when $\mu^*$ is close to~$1$:
\begin{corollary}
    When $\mu^* \geq 1-\delta$, there is a Best Observable Selection algorithm with expected regret at most $\sqrt{\frac{2 \ln m}{n}} \sqrt{\delta} + \frac{2\ln m}{n}$.
\end{corollary}
\begin{proof}
    We repeat the earlier proof to get to $\dKL{\ba^*}{\wt{\ba}} \leq \frac{\ln m}{n}$, where $\ba^*$ and $\wt{\ba}$ are Bernoulli.  Then $\EE[\ba^*] = \mu_* \geq 1-\delta$ and a  version of Pinsker (e.g.~\cite[Lemma~6]{GMS19}) indeed implies $\EE[\wt{\ba}] \geq \EE[\ba^*] - (\sqrt{\frac{2\ln m}{n}}\sqrt{\delta} + \frac{2\ln m}{n})$.
\end{proof}

\noindent Our proof of Theorem~\eqref{thm:BOS} can be viewed as a kind of quantization of the Thompson-sampling result of Russo--Van{ }Roy~\cite{russo2016information}, which in turn uses the information theory methods of Xu--Raginsky~\cite{xu2022minimum} and the minimax idea of Hafez-Kolahi et al.~\cite{hafezkolahi2023minimax}.
\end{document}